\documentclass[twoside]{article}

\usepackage[accepted]{aistats2024}
\usepackage{hyperref}       
\usepackage{url}            
\usepackage{booktabs}       
\usepackage{amsfonts}       
\usepackage{nicefrac}       
\usepackage{microtype} 
\usepackage{booktabs}
\usepackage{multirow}
\usepackage{graphicx}
\usepackage{dsfont}
\usepackage{bm}
\usepackage{xcolor}         
\usepackage{soul}
\usepackage{color}
\usepackage{lipsum}
\usepackage{tabularx,booktabs}
\usepackage{caption}
\usepackage{subcaption}
\usepackage{enumitem}
\usepackage{framed}
\usepackage{float}
\usepackage{wrapfig}
\usepackage{adjustbox}
\usepackage{amsmath}
\usepackage{amsthm}
\usepackage{verbatim}
\usepackage{dblfloatfix}
\usepackage{colortbl}

\usepackage[linesnumbered,ruled,vlined]{algorithm2e}
\usepackage{mathtools}
\usepackage{algpseudocode}

\newcommand{\Prob}[1]{\Pr\left[#1\right]}

\newcommand{\getsr}{{\:\stackrel{{\scriptscriptstyle \hspace{0.2em}\$}} {\leftarrow}\:}}

\newcommand{\getsp}{{\:\stackrel{{\scriptscriptstyle \hspace{0.2em}\mathcal{D}}} {\leftarrow}\:}}
\newcommand{\nosemic}{\renewcommand{\@endalgocfline}{\relax}}
\newcommand{\dosemic}{\renewcommand{\@endalgocfline}{\algocf@endline}}
\let\oldnl\nl
\newcommand{\nonl}{\renewcommand{\nl}{\let\nl\oldnl}}
\newcommand{\Adv}{{\mathbf{\bf Adv}}}

\newcommand{\experimentv}[1]{\underline{{#1}}\smallskip}
\newcommand{\bits}{{\{0,1\}}}

\newcommand{\datt}{{d_{\textup{attn}}}}

\newcommand{\cmt}[1]{{\color{blue} \footnotesize \# #1}}

\newtheorem{theorem}{Theorem} 
\newtheorem*{theorem*}{Theorem}
\newtheorem{lemma}{Lemma}
\newtheorem*{lemma*}{Lemma}

\newtheorem{definition}{Definition}
\newtheorem{remark}{Remark}

\usepackage[round]{natbib}

\begin{document}

%

%

\twocolumn[

\aistatstitle{Analysis of Privacy Leakage in Federated Large Language Models}

\aistatsauthor{Minh N. Vu \And Truc Nguyen \And  Tre’ R. Jeter \And My T. Thai}

\aistatsaddress{ University of Florida \And University of Florida \And University of Florida \And University of Florida  } ]

\begin{abstract}
  With the rapid adoption of Federated Learning (FL) as the training and tuning protocol for applications utilizing  Large Language Models (LLMs), recent research highlights the need for significant modifications to FL to accommodate the large-scale of LLMs. While substantial adjustments to the protocol have been introduced as a response, comprehensive privacy analysis for the adapted FL protocol is currently lacking.

To address this gap, our work delves into an extensive examination of the privacy analysis of FL when used for training LLMs, both from theoretical and practical perspectives. In particular, we design two active membership inference attacks with guaranteed theoretical success rates to assess the privacy leakages of various adapted FL configurations. Our theoretical findings are translated into practical attacks, revealing substantial privacy vulnerabilities in popular LLMs, including BERT, RoBERTa, DistilBERT, and OpenAI's GPTs, across multiple real-world language datasets. Additionally, we conduct thorough experiments to evaluate the privacy leakage of these models when data is protected by state-of-the-art differential privacy (DP) mechanisms. 
\end{abstract}

\section{INTRODUCTION}

Recent years have observed the exceptional capabilities of Large Language Models (LLMs)~\citep{taylor2022galactica} in many complex real-world applications, especially those involving AI-generated conversations of ChatGPT~\citep{chatgpt} and Google Bard~\citep{google_bard_2023}. This has spurred significant research efforts in training and utilizing LLMs in various important domains across multiple industries~\citep{huang2023finbert, wu2023bloomberggpt, PMID:37438534,liu2023using}. Nevertheless, the success of LLMs comes at the cost of massive amounts of training data as well as computational resources. In response, many recent research~\citep{ adapter_tuning, hu2022lora, li-liang-2021-prefix, ben-zaken-etal-2022-bitfit} independently propose the usage of Federated Learning (FL)~\citep{McMahan2016CommunicationEfficientLO, FLBook} to resolve those challenges as it allows the utilization of private data and distributed resources. 

Due to the massive size of LLMs, directly training or fine-tuning them for down-stream tasks on FL would cause substantial communication overhead and place heavy burdens on the storage and computational resources of participating devices~\citep{hu2022lora, zhang-etal-2023-fedpetuning, chen2023federated}. A natural solution for the issues, known as parameter-efficient training and tuning (PET)~\citep{zhang-etal-2023-fedpetuning}, is to update only a small number of the parameters while freezing the rest (Fig.~\ref{fig:petuning}). The PET methods either inject some additional trainable parameters~\citep{li-liang-2021-prefix},  introduce some extra layers~\citep{adapter_tuning, hu2022lora}, or update only some portions of the original LLMs~\citep{ben-zaken-etal-2022-bitfit}. These studies have shown the modified FL can achieve comparable performance to traditional FL; nevertheless, a critical gap remains in terms of privacy analysis for these novel protocols. While the overall reduction in communication messages can potentially lower the attack surface, it remains uncertain whether the modified FL methods offer enhanced security as the introduction of extra layers and parameters may introduce new vulnerabilities.


 \begin{figure}[ht]
     \centering
         \includegraphics[width=0.74\linewidth]{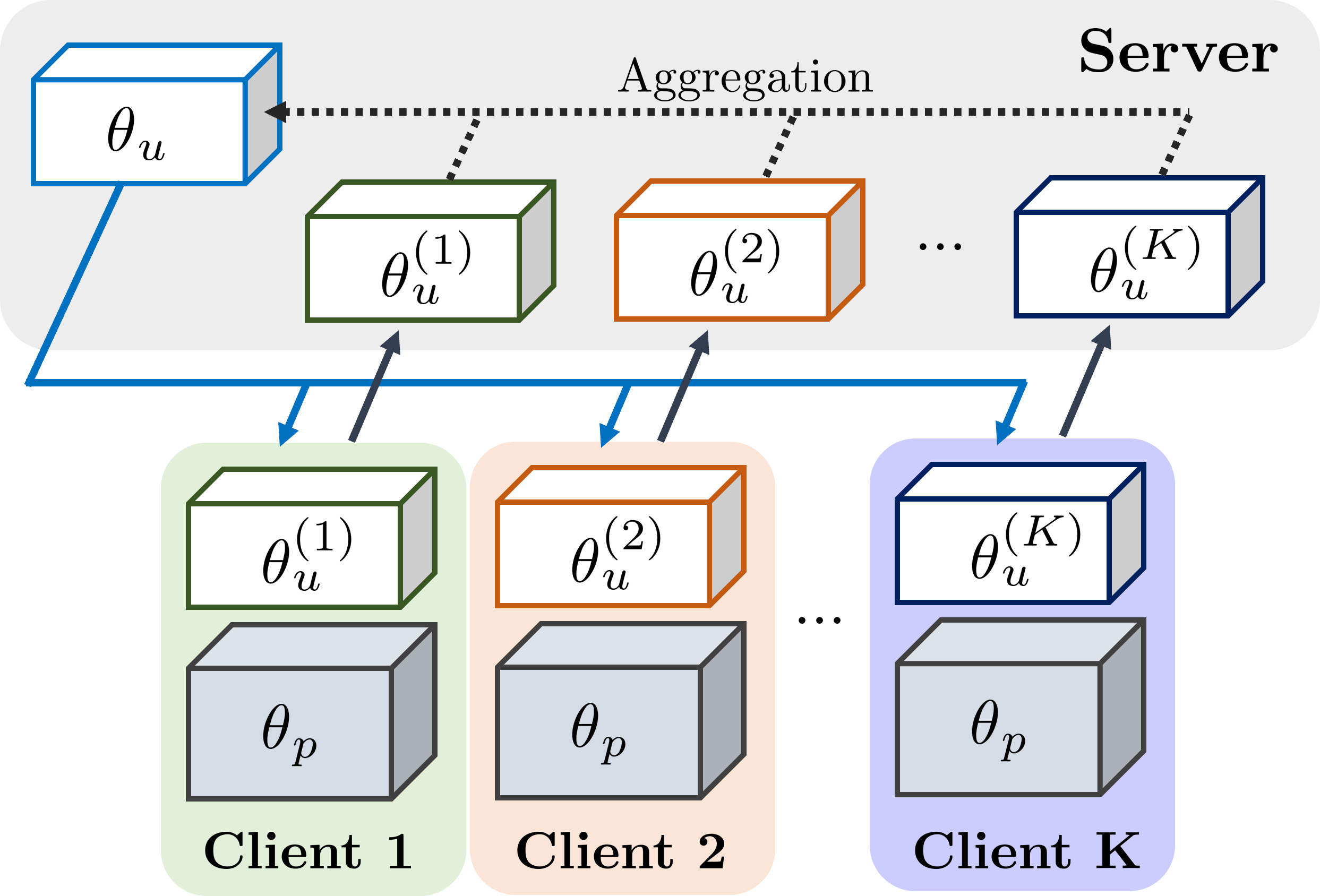}
     \caption{Training/tuning LLMs in FL: The clients typically
exchange a light amount of trainable parameters $\theta_i$ while keeping most parameters, i.e., $\theta_p$, frozen.} 
     \label{fig:petuning}
 \end{figure}

Motivated by that lack of study, this paper investigates the privacy leakage in the adoption of FL to LLMs from both theoretical and practical perspectives. Our study focuses on the active setting, wherein the FL server operates dishonestly by manipulating the trainable weights to compromise privacy. Our findings reveal that clients' data are fundamentally vulnerable to active membership inference (AMI) carried out by a dishonest server. The theme of our study is the construction of low-complexity adversaries with provable high attack success rates. As we will demonstrate in our theoretical results, establishing complexity bounds for these adversaries is imperative as it forms the foundation for rigorous security/privacy statements about the protocol. The challenges are not only in determining the sweet spots between adversaries' complexity and the attacking success rates but also in showing how the theoretical vulnerabilities manifest as practical risks in the context of FLs for LLMs.
For that purpose, our attacks are designed to exploit the trainable fully connected (FC) layers and self-attention layers in FL updates as both are widely adopted to utilize LLMs. 
Our main contributions are summarized as follows:

\begin{itemize}
    \item We prove that the server can exploit the FC layers to perfectly infer membership information of local training data (Theorem~\ref{theorem:ami_none}) in FL with LLMs. 
    \item When the trainable weights belong to a self-attention layer, we introduce a sefl-attention-based AMI attack with a significantly high guarantee success rate and demonstrate a similar privacy risk as in the case of FC layers (Theorem~\ref{theorem:ami_attn}). 
    \item Practical privacy risks of utilizing LLMs in FL to AMI are demonstrated through the implementations of the formulated adversaries. Our experiments are conducted on five state-of-the-art LLMs, including BERT-based~\citep{BERT2019, Liu2019RoBERTaAR, sanh2019distilbert} and GPT models~\citep{gpt2, gpt3} on four different datasets. We heuristically assess the privacy risks for both unprotected and Differential Privacy protected data.
    
    
\end{itemize}

\textbf{Organization.} Sect.~\ref{sect:prelim} discusses the background, the related works, and the notations used in this manuscript. Sect.~\ref{sect:threat_models} describes the AMI threat models. Sects.~\ref{sect:ami_mlp} provide the descriptions of our attacks and their theoretical analysis. Sect.~\ref{sect:experiments} reports our experimental results and Sect.~\ref{sect:conclude} concludes this paper. 

\section{BACKGROUND, RELATED WORKS, AND NOTATIONS} \label{sect:prelim}
This section provides the background, related works, and notations that we use in this work.

\textbf{Federated Learning (FL).} FL~\citep{FLBook} is a collaborative learning framework that allows model updates across decentralized devices while keeping the data localized. The training is typically orchestrated by a central server. 
At the beginning, the server initializes the model's parameters. For each training iteration, a subset of clients is selected to participate. Each of them computes the gradients of trainable parameters on its local data. The gradients are then aggregated among the selected clients. The model's parameters are then updated based on the aggregated gradients. The training continues until convergence.

 \textbf{Parameter-Efficient Training and Tuning.} There has been a notable increase in research on PET to avoid full model updates in FL. Generally, these methods focus on updating lightweight trainable parameters while keeping the rest frozen. For instance, the adapter approach~\citep{adapter_tuning,pfeiffer-etal-2020-adapterhub} inserts two adapter layers to each Transformer block~\citep{Vaswani2017}. The computation can be described formally as $h \leftarrow W_u \sigma (W_d h)$, where $\sigma$ is the nonlinear activation function. Reparameterization-based methods such as LoRA~\citep{hu2022lora} aim to optimize the low-rank decomposition for weight update matrices. On the other hand, BitFit~\citep{ben-zaken-etal-2022-bitfit} empirically demonstrates that only updating the bias terms can still lead to competitive performance. Another common approach is prompt-tuning~\citep{li-liang-2021-prefix}  which attaches trainable vectors, namely prompt, to the model's input. 
 

 \textbf{Attention mechanism in FL}. The Attention mechanism is a machine learning technique that allows a model to focus on specific parts of the input. It has been widely used in machine translation, video processing, speech recognition, and many other applications. Some popular models using attention mechanisms are Transformer~\citep{Vaswani2017}, BERT~\citep{BERT2019}, RoBERTa~\citep{Liu2019RoBERTaAR}, and GPTs~\citep{gpt2,gpt3}. With its popularity, it is increasingly common to find attention mechanisms in models trained in FL settings, e.g. models of Google~\citep{fl_att_google, fl_att_google_2, fl_att_google_3}, Amazon~\citep{fl_att_amazon}, and many others~\citep{stremmel2021pretraining, Chen2020FederatedML, bui2019federated}. 

\textbf{FL with Differential Privacy.} Differential Privacy (DP)~\citep{dwork2006calibrating, erlingsson2014rappor} has been recognized as a solution to 
mitigate the privacy risks of gradient sharing in FL. The principle of DP is based on \textit{randomized response}~\citep{warner1965randomized}, initially introduced to maintain the confidentiality of survey respondents. The definition of $\varepsilon$-DP is:

\begin{definition}{$\varepsilon$-DP.} A randomized algorithm $\mathcal{M}$ fulfills $\varepsilon$-DP, if for all subset $ \mathcal{S} \subseteq \textup{Range}(\mathcal{M})$  and
for any adjacent datasets $D$ and $D'$, we have:
$Pr[\mathcal{M}(D) \in \mathcal{S}] \leq e^{\varepsilon} Pr[\mathcal{M}(D') \in \mathcal{S}]$,
where $\varepsilon$ is a privacy budget.
\label{Different Privacy} 
\end{definition}

The privacy budget $\varepsilon$ controls the privacy level: a smaller $\varepsilon$ enforces a stronger privacy guarantee; however, it reduces data utility as the distortion is larger.


\textbf{Related works.} The first AMI attack by a dishonest server in FL was recently introduced by~\cite{NasrSH19}. The attack relies on multiple model updates for inference. Later, the work~\citep{nguyen2023active} introduces a stronger and stealthier attack requiring only one FL iteration; however, a separate neural network trained on the dataset is needed. Both attacks have non-trivial time complexity and provide no theoretical guarantees. Furthermore, neither of these works is tailored for FL with LLMs, and their efficacy in this context has not been validated.

During our research, we noticed related privacy concerns regarding FL with LLMs. The work of~\cite{Wang2023CanPL} and~\cite{yu2023federated} explore differentially private FL with LLMs, primarily focusing on assessing the model's performance while assuming privacy budgets. In contrast, our goal is to analyze new attack vectors in modified FL for LLMs through novel adversarial inference schemes.

Another line of related work focuses on inference attacks, where the attackers aim to deduce the private features of clients in FL~\citep{Hitaj2017DeepMU, NEURIPS2019_60a6c400, 8953573, featureinferenceattack2021_cafe}. While inference attacks can be considered as more potent than membership inferences since they can recover the entire input, their effectiveness heavily relies on the training model and method-specific optimizations. These dependencies hinder the establishment of theoretical guarantees of inference attacks. In contrast, membership inferences enable us to circumvent these dependencies, a crucial aspect for establishing Theorems~\ref{theorem:ami_none} and~\ref{theorem:ami_attn}, which provide formal statements regarding the vulnerability of Federated LLMs.



\textbf{Notations.} We consider the private data $D$ in LLMs consists of 2-dimensional arrays, and write $D = \{X_i\}_{i=1}^n$ for $X_i \in \mathcal{X}$ where $\mathcal{X} \subseteq \mathbb{R}^{d_X\times l_X}$. We use $\mathcal{D}$ to denote the data distribution. Each column of a 2-dimensional array $X$, denoted by $x_j \in \mathbb{R}^{d_X}$, is referred as a \textit{token}. We also denote $M$ as the largest $L_2$ norm of the tokens, i.e., $M = \max_{X \in D} \max_{x_j \in X} \| x_j\|$.

Regarding the self-attention mechanisms, for an input $X \in \mathbb{R}^{d_X \times l_X}$, the layer's output $Z^h \in \mathbb{R}^{d_{\textup{hid}} \times l_X}$ of the attention head $h$ is given by:
\begin{align}
    Z^h &=  W_V^h X \textup{softmax} \left( \nicefrac{1}{\sqrt{d_{\textup{attn}}}} \ X^\top {W_K^{h \top}} W_Q^h X \right) \label{eq:attn_weight}
\end{align}
where $W_Q^h \in \mathbb{R}^{\datt \times d_X}, W_K^h \in \mathbb{R}^{\datt \times d_X}$ and $W_V^h \in \mathbb{R}^{d_{\textup{hid}} \times d_X}$ are the trainable weights of the head $h$. 
The output of the layer after ReLU activation is $ Y =   \textup{ReLU} \left(\sum\nolimits_{h=1}^H W_O^h Z^h + b_O 1^\top \right)$
where $H$ is the number of heads, $W_O^h  \in \mathbb{R}^{d_Y \times d_{\textup{hid}}  }$ are the trainable weights and $b_O \in \mathbb{R}^{d_Y}$ are the trainable biases for the aggregation of the layer's final output. 

\section{THE THREAT MODEL} \label{sect:threat_models}

As the FL server defines the model's architecture and distributes the parameters, it can deviate from the protocol to strengthen the privacy attacks~\citep{boenisch2021curious, nguyen2022blockchain, fowl2021robbing}. We formalize our study on the security of FL with LLMs via a security game denoted by \textit{AMI Security Game} in Subsect.~\ref{subset:ami_security}. In the game, a dishonest server maliciously specifies the model's architecture and modifies its parameters to infer information about the local training data of a client. We then discuss how the AMI security game can capture the security threat of FL with LLMs in Subsect.~\ref{subset:ami_game_w_LLMs}.



\subsection{AMI Security Games} \label{subset:ami_security}

\begin{figure}[ht]
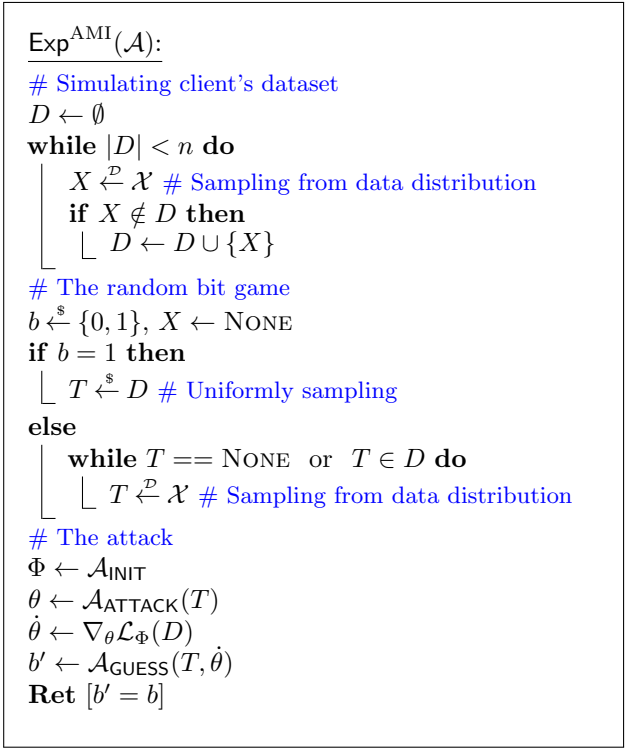

\begin{framed}
\begin{flushleft}
\experimentv{$\mathsf{Exp}^{\textup{AMI}}{(\mathcal{A})}$:}  \\
\cmt{Simulating client's dataset} \\ 
$D \leftarrow \emptyset$  \\ 
\While{$|D| < n$} 
{$X \getsp \mathcal{X} $ \cmt{Sampling from data distribution} \\
\If{$X \notin D$}{$D \leftarrow D \cup \{ X \}$ }}
\cmt{The random bit game}\\
$b \getsr \bits$, $X \leftarrow \textsc{None}$ \\
\If{$b=1$} {
$T \getsr D$ \cmt{Uniformly sampling}  \\
}
\Else{
\While{$T == \textsc{None}$ { \textup{or} } $T \in D$}
{
$T \getsp \mathcal{X}$ \cmt{Sampling from data distribution} \\
}
}
\cmt{The attack}\\
$\Phi \gets \mathcal{A}_{\mathsf{INIT}}$ \\
$\theta \gets \mathcal{A}_{\mathsf{ATTACK}}(T)$ \\
$ \dot{\theta} \gets \nabla_\theta \mathcal{L}_{\Phi}(D)$ \\
$b' \gets \mathcal{A}_{\mathsf{GUESS}}(T, \dot{\theta})$ \\
\textbf{Ret} $[b'=b]$ 
\end{flushleft}
\end{framed}
\caption{The AMI Threat Model as a Security Game.}
\label{fig:AMI}
\end{figure}

We formalize the AMI threat models as in the standard settings of existing works~\citep{shokri2017membership, carlini2022membership,yeom2018privacy} into the security games $\mathsf{Exp}^{\textup{AMI}}(\mathcal{A})$ with description in Fig.~\ref{fig:AMI}. The adversarial server $\mathcal{A}$ in the security games consists of 3 components $\mathcal{A}_{\mathsf{INIT}}$, $\mathcal{A}_{\mathsf{ATTACK}}$ and $\mathcal{A}_{\mathsf{GUESS}}$. First, a randomly generated bit $b$ is used to determine whether the client's data $D$ contains a target sample $T$. If the protocol allows the server to decide the training architecture, the first step of the server $\mathcal{A}_{\mathsf{INIT}}$ decides a model $\Phi$ for the training. If that is not the case, $\mathcal{A}_{\mathsf{INIT}}$ simply collects information on the model from the protocol. Then, $\mathcal{A}_{\mathsf{ATTACK}}$ crafts the model's parameters $\theta$ based on the target $T$ and the trained architecture. Upon receiving $\Phi$ and $\theta$, the client computes the gradients $\dot{\theta} = \nabla_\theta \mathcal{L}_{\Phi}(D)$ and sends them to the server, where $\mathcal{L}$ denotes a loss function for training. With $\dot{\theta}$, $\mathcal{A}_{\mathsf{GUESS}}$ guesses the value of $b$. Correctly inferring $b$ is equivalent to determining whether $T$  is in the local data $D$. The advantage of the adversarial server $\mathcal{A}$ in the security game is given by:
 \begin{align*}
      &\Adv^{\textup{AMI}}(\mathcal{A}) = 2 \Pr[\mathsf{Exp}^{\textup{AMI}}(\mathcal{A}) = 1] - 1 \\
      &= \Pr[b'=1|b=1] + \Pr[b'=0|b=0] - 1 
\end{align*}
 where $\Pr[b'=1|b=1]$ and $\Pr[b'=0|b=0]$ are the True Positive Rate and the True Negative Rate of the adversary, respectively. The existence of an adversary with a high advantage implies a high privacy risk/vulnerability of the protocol described in the security game. 

\subsection{AMI Security Games for Federated LLMs} \label{subset:ami_game_w_LLMs}


To ensure $\mathsf{Exp}^{\textup{AMI}}(\mathcal{A})$ truly captures practical threats in Federated LLM, additional constraints are needed. For example, to describe PET, $\mathcal{A}_{\mathsf{ATTACK}}$ should exclusively target trainable weights $\theta_u$ at specific model locations. We now describe those conditions and our proposed attacks will align accordingly.
 \begin{figure}[ht]
     \centering
         \includegraphics[width=0.85\linewidth]{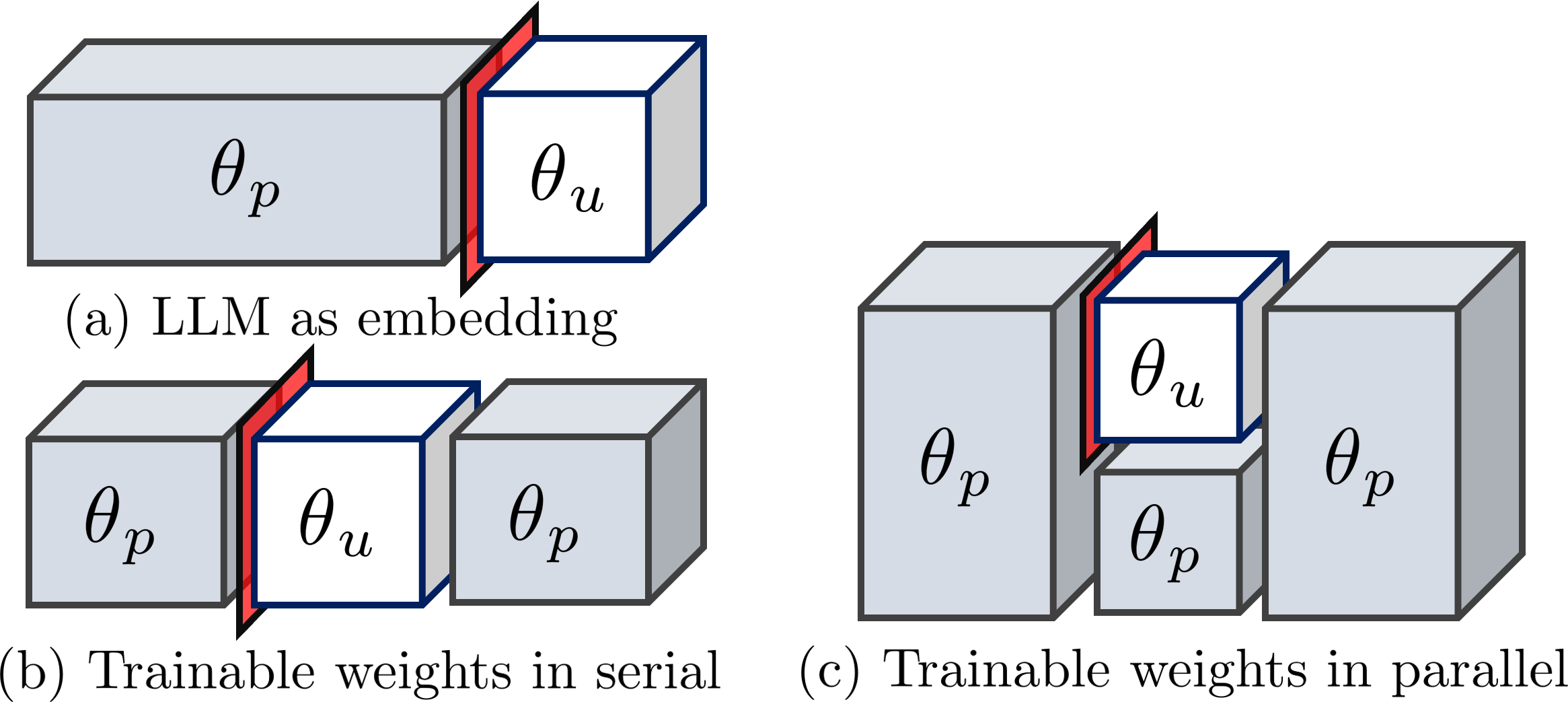}
     \caption{Different scenarios in training/fine-tuning LLMs in FL. The red squares show the privacy leakage surfaces in the threat model. The white and grey boxes indicate the trainable and frozen weights, respectively.} 
     \label{fig:petuning_2}
 \end{figure}

Fig.\ref{fig:petuning_2} illustrates typical FL update scenarios with LLMs and their associated privacy leaking surfaces. The most common configuration (a) is when the LLM or some of its layers are used as embedding modules. In the second and third scenarios, additional modules with trainable weights are introduced sequentially (b) or in parallel (c). They capture the PET strategies involving adapters~\citep{adapter_tuning, pfeiffer-etal-2020-adapterhub} and reparameterization tricks~\citep{hu2022lora}. 

To make $\mathsf{Exp}^{\textup{AMI}}(\mathcal{A})$ describe the threats in Federated LLMs, the type of layers and weights in the trainable modules need to be specified. The assumption is that their inputs are hidden representations of the input data at some fixed locations of the language models. The most common layers are the fully connected (FC) layers and the attention mechanism~\citep{Vaswani2017}. Consequently, our inference attacks are designed for those layers. 

It is noteworthy to point out that, the actual pattern that the adversaries operate on is the embedded version of $T$, i.e., $\Phi_{\theta_p}(T)$. As the pre-trained parameters $\theta_p$ are public, and under a mild assumption that different target $T$ result in different embedding $\Phi_{\theta_p}(T)$, the problems of inferring $T$ and  $\Phi_{\theta_p}(T)$ are equivalent~\footnote{For LLMs, the assumption means different input texts result in different embedding, which is quite reasonable.}. Therefore, in the subsequent discussions, we treat the embeddings as the user's data for ease of notation.



\section{ACTIVE MEMBERSHIP INFERENCE ATTACKS} \label{sect:ami_mlp}
This section presents our membership attacks and their theoretical guarantees in inferring data in FL on LLMs. The first attack, called FC-based adversary $\mathcal{A}_{\mathsf{FC}}$,  is designed for scenarios where the first two trainable layers (See Fig.~\ref{fig:petuning_2}) are FC layers. The second attack, Attention-based adversary $\mathcal{A}_{\mathsf{Attn}}$, is tailored for cases when the first trainable layer is self-attention. The primary goal of both attacks is to create a neuron in those layers such that {it is activated if and only if the target pattern is fed to the model}. Consequently, the gradients of the weights computing that neuron are non-zero if and only if the target pattern is in the private training batch.

\begin{figure*}[h]
     \centering
         \includegraphics[height=0.12\linewidth]{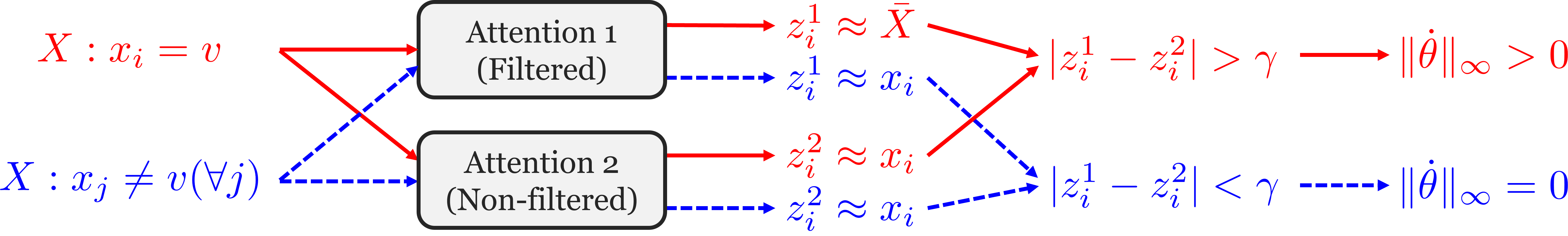}
     \caption{The $\mathcal{A}_{\mathsf{Attn}}$ adversary exploiting self-attention mechanism for membership inference in FL: If the target pattern $v= x_i$ is in the data, the output $z^1_i$ of the filtered head approximates the token's average $\Bar{X}$ instead of approximating $x_i$. This creates non-zero gradients for weights computing the difference between two heads.}
     \label{fig:attack_principle}
 \end{figure*}

    \subsection{FC-based adversary} \label{subsect:fc_adv}
    
    The FC-based adversary $\mathcal{A}_{\mathsf{FC}}$ operates as follows. First, if the protocol permits, the initialization $\mathcal{A}_{\mathsf{FC-INIT}}$ specifies the model $\Phi$ in the FL training to use FC as its first two layers. For a target pattern $T$ with a dimension of $d_T$,  the weights $W_1$ and biases $b_1$ of the first layer are set to dimensions $2d_T \times d_T$ and $2d_T$, respectively. If the protocol requires a larger dimension, extra parameters can be ignored. Conversely, if only smaller dimensions are allowed, the attacker uses a substring of $T$ as the detected pattern and proceeds from there. Since $\mathcal{A}_{\mathsf{FC}}$ uses only one output neuron at the second FC layer, the attack only needs to set one row and one entry of its weights and biases, denoted by $W_2[1,:]$ and $b_2[1]$. Particularly, the weights and biases of the two FC layers are set by $\mathcal{A}_{\mathsf{FC-ATTACK}}$ as:
     \begin{align}
         W_1 \leftarrow \begin{bmatrix}
        I_{d_T} \\ 
        - I_{d_T}
        \end{bmatrix}, \quad  b_1 \leftarrow \begin{bmatrix}
        - T \\ 
        T 
        \end{bmatrix} \nonumber \\
        W_2[1,:] \leftarrow -1_{d_T}^\top, \quad   b_2[1] \leftarrow \tau \label{eq:tau_first}
     \end{align}
     where $I_{d_T}$ is the identity matrix and $1_{d_T}$ is the one-vector of size $d_T$. The parameter $\tau$ controls the allowable distance between an input $X$ and the target $T$, which can be obtained from the distribution statistics. In the guessing phase, $\mathcal{A}_{\mathsf{FC-GUESS}}$ returns $1$ if the gradient of $b_2[1]$ is non-zero and $0$ otherwise. Appx.~\ref{appx:ami_adv} provides a more detailed description of this attack.

     \textbf{Attack strategy.} Upon receiving an input $X$, the two FC layers compute $z_0 :=\max \left \{ b_2[1]  - \| X - T\|_{L_1}, 0 \right \}$. If there is an $X = T$ in the data, $z_0$ will be activated and the gradient of the bias $b_2[1]$ is non-zero. On the other hand, for $b_2[1]=\tau > 0$ small enough, $z_0 = 0$ for all $X \neq T$. Thus, the gradient of the bias $b_2[1]$ is zero when $T \notin D$. Therefore, the gradient of $b_2[1]$ indicates the presence of $T$ in the local data.

    \begin{remark}\label{remark:twoFC}\textbf{The dimension of the target $T$.}
    In practice, it is uncommon to feed the embedding of the entire input to FC layers. Instead, it is more common to forward each token individually. This means the target $T$ is the embedding of a token with a dimension of $d_X$. We might expect that this could hinder adversarial inference as a smaller portion of the input is exploited. However, as LLMs create strong connections among tokens, each token can embed the signature of the entire sentence. While this characteristic benefits the model's performance, it also enables the inference of the input from the token level. Our experiments (Sect.~\ref{sect:experiments}) will consider both sentence embedding and token embedding and illustrate the above claim. 
\end{remark}

     \textbf{Attack advantage.} The attack strategy implies that the adversary wins the security game $\mathsf{Exp}^{\textup{AMI}}$ with probability 1. We formalize that claim in  Lemma~\ref{lemma:ami_none}:
     \begin{lemma} \label{lemma:ami_none}
    The advantage of the adversary $\mathcal{A}_{\mathsf{FC}}$ in the security game $\mathsf{Exp}^{\textup{AMI}}$ is 1, i.e., $\Adv^{\textup{AMI}}(\mathcal{A}_{\mathsf{FC}})  = 1$. (Proof in Appx.~\ref{appx:ami_mlp_noldp})
    \end{lemma}
    Since  $\mathcal{A}_{\mathsf{FC}}$ can be constructed in $\mathcal{O} (d_T^2)$, the Lemma gives us the following theoretical result on the vulnerability of unprotected data in FL against AMI:
    \begin{theorem} \label{theorem:ami_none}
    There exists an AMI adversary $\mathcal{A}$ exploiting 2 FC layers with time complexity $\mathcal{O} (d_T^2)$  and advantage $\Adv^{\textup{AMI}}(\mathcal{A})  = 1$ in the game $\mathsf{Exp}^{\textup{AMI}}$.
    \end{theorem}
     
The implication of Theorem~\ref{theorem:ami_none} is that unprotected private data of clients in FL is exposed to very high privacy risk, which was also observed by~\cite{NasrSH19, nguyen2023active}. However, the adversaries in those works involve multiple updates of either the FL model or other neural networks; thus, theoretical results on the trade-off between their success rates and complexities are not yet available. Therefore, it is unknown whether the vulnerability of Federated LLMs can be rigorously established from those attacks.

\begin{remark}\label{remark:FC_assumption}\textbf{Attack assumptions.}
    {The FC attack $\mathcal{A}_{\mathsf{FC}}$ does not need any distributional information to work on unprotected data. In fact, the attacker just needs to specify $\tau$ (\ref{eq:tau_first}) small enough such that $\tau < \| X_1 - X_2 \|_{L_1}$ for any $X_1 \neq X_2$ in the model's dictionary. This is shown in the proof of Lemma~\ref{lemma:ami_none} (Appx.~\ref{appx:ami_mlp_noldp}). Since the dictionary or the tokenizer is public, selecting $\tau$ does not require any additional information.}
\end{remark}


\subsection{Attention-based adversary} \label{subsect:attn_adv}

Our proposed $\mathcal{A}_{\mathsf{Attn}}$ exploits the memorization capability of the self-attention, which was indirectly studied by~\cite{ramsauer2021hopfield}. That work shows the attention is equivalent to the proposed \textit{Hopfield} layer, whose main purpose is to directly integrate memorization into the layer. We adopt that viewpoint and introduce a specific configuration of the attention to make it memorize the local training data while filtering out the target of inference. 


 Since $\mathcal{A}_{\mathsf{Attn}}$ operates at a token level, the target of inference is the token resulting from embedding the target $T$, denoted by $v \in \mathbb{R}^{d_X}$. The intuition of $\mathcal{A}_{\mathsf{Attn}}$ is shown in Fig.~\ref{fig:attack_principle}: by setting an attention head to memorize the input batch and filter out the target token, $\mathcal{A}_{\mathsf{Attn}}$ introduces a gap between that head's output and the output of a non-filtered head. The gap is then exploited to reveal the victim's data. The following describes the components, strategy, and advantage of this attack.

\textit{Initialization $\mathcal{A}_{\mathsf{Attn-INIT}}$:} The attack uses 4 attention heads, i.e., $H$ to 4. The layer dimensions are $d_{\textup{attn}} = d_X -1$, $d_{\textup{hid}} = d_X$ and $d_Y = 2 d_X$. Any configurations with more parameters can adopt this attack since the extra parameters can simply be ignored. 
    \SetKwInput{KwInput}{Hyper-parameters}
    \IncMargin{1.3em}
\begin{algorithm}[ht]
    \SetAlgoLined
    \KwInput{$\beta, \gamma \in \mathbb{R}^+$} 
     Randomly initialize $W_Q^h, W_K^h, W_V^h, W_O $ and $b_O$ for all head $h \in \{1,2,3,4 \}$
    
    Randomly initialize a matrix $W \in \mathbb{R}^{d_X \times d_X}$ 
    \label{line:init}
    
    $W[:,1] \leftarrow v$ \cmt{Set the first column of $W$ to $v$} \label{line:setv}
    
    $Q, R \leftarrow \textup{QR}(W)$ \cmt{QR-factorization $W$}  \label{line:QR}
    
    $W_Q^1 \leftarrow Q[2:d_X]^\top$ \cmt{Embed pattern to $W_Q^1$} \label{line:Q_assign}
    
    $W_K^{1} \leftarrow \beta {{W_Q^{1}}^{\dagger \top}} $ \cmt{Set head $1$ to memorization} \label{line:pinv1}
    
    $W_K^{2} \leftarrow \beta {{W_Q^{2}}^{\dagger \top}} $ \cmt{Set head $2$ to memorization}  \label{line:pinv2}
    
    $W_Q^{3} \leftarrow W_Q^{1}$,  $W_K^{3} \leftarrow W_K^{1}$  \cmt{Copy head 1 to head 3}
    
    $W_Q^{4} \leftarrow W_Q^{2}$,  $W_K^{4} \leftarrow W_K^{2}$  \cmt{Copy head 2 to head 4}
    
    $W_V^{1} \leftarrow I_{d_X}$,
    $W_V^{2} \leftarrow I_{d_X}$, $W_V^{3} \leftarrow I_{d_X}$, $W_V^{4} \leftarrow I_{d_X}$ 
    
    $W_O \leftarrow \begin{bmatrix}
    I_{d_X} & -I_{d_X} & 0_{d_X} & 0_{d_X}\\ 
    0_{d_X} & 0_{d_X} & -I_{d_X} & I_{d_X}
    \end{bmatrix}$ 
    
    ${b_O}_i = - \gamma, \forall i \in \{1,\cdots, d_Y\}$ 
    
    \textbf{Ret} all weights and biases

\caption{$\mathcal{A}_{\mathsf{Attn-ATTACK}}(v)$}
\label{algo:api_attack}
\end{algorithm}

 \textit{Attack $\mathcal{A}_{\mathsf{Attn-ATTACK}}$:} 
    This step (Algo.~\ref{algo:api_attack}) sets the attention weights $W_Q^h, W_K^h, W_V^h, W_O $, and bias $b_O$, where $h$ is the head's index. There are two hyper-parameters, $\beta$ and $\gamma \in \mathbb{R}^+$. While $\beta$ controls how much the heads memorize the input patterns $x_i^h$, $\gamma$ adjusts a cut-off thresholding between $v\in D$ and $v \notin D$ (Fig.~\ref{fig:attack_principle}). Given a target $v$, the first attention head is set such that:
    \begin{align}
        & {W_K^{1 \top}} W_Q^1 \approx \beta  I_{d_X}  \quad \textup{and} \quad W_Q^1  v \approx 0 \label{eq:api_conds}
    \end{align}
   To enforce (\ref{eq:api_conds}), $d_X-1$ vectors orthogonal to $v \in \mathbb{R}^{d_X}$ are assigned to  $W_Q^1$ via QR-factorization (line \ref{line:setv}-\ref{line:Q_assign}). Then, $W_K^1$ is set to the transpose of $\beta {{W_Q^{1 \dagger}}}$, where $\dagger$ denotes the pseudo-inverse. On the other hand, the second head randomizes $W_Q^2$ and assigns $W_K^2$ as its pseudo-inverse. This means the second condition of (\ref{eq:api_conds}) does not hold for the second head. The other parameters of the first two heads are set so that the first $d_X$ rows of $Y$ compute $\max\{0, Z^1 - Z^2 - \gamma 1^\top\}$. The third and the fourth heads are configured to compute the negations of the first two, i.e., they make the last $d_X$ rows of $Y$ return $\max\{0, Z^2 - Z^1 - \gamma 1^\top\}$. For ease of analysis, we construct $W_V^h$ and $W^O$ from identity and zero matrices. 
    
\textit{Guessing $\mathcal{A}_{\mathsf{Attn-GUESS}}$:} The attacker checks whether any of the weights in $W_O$ have non-zero gradients, and returns $b'=1$ if that is the case.

\textbf{Attack strategy.} 
     $\mathcal{A}_{\mathsf{Attn}}$ exploits the memorization imposed by the first condition of (\ref{eq:api_conds}). To see how it works, we consider following the two cases:
      
     Case 1: If $v \notin X$, $\nicefrac{1}{\beta}X^\top {W_K^{1 \top}} W_Q^1 X \approx X^\top X $, which is the correlation matrix of the tokens. The softmax's output then approximates $I_{l_X}$ as the diagonal of  $X^\top X $ is larger than other entries. The head's output $Z^1\approx X$, i.e., $z^1_i \approx x_i$, as a result. Since the second head behaves similarly in this case, we have $Z^2\approx X$ and $z^2_i \approx x_i$.
     
         Case 2: When $X$ contains $v$, the second condition of (\ref{eq:api_conds}) makes $x_i^\top {W_K^{1 \top}} W_Q^1 x_i \approx 0$, and consequently causes the softmax's output uniform, i.e., the attention is distributed equally among all tokens. Thus, the first head's output is the token's average $z^1_i \approx \Bar{X}$. Since the second head does not filter $v$, $ z^2_i \approx x_i$. The difference $|z^1_i - z^2_i |$ then reveals the presence of $v$ in $X$.

\textbf{Attack advantage.} The advantage of $\mathcal{A}_{\mathsf{Attn}}$ depends on an intrinsic measure of the data, called the \textit{Separation of Patterns}~\citep{ramsauer2021hopfield}:
\begin{definition} (Separation of Patterns).
For a token $x_i$ in $X = \{x_j\}_{j=1}^{l_X}$, its separation $\Delta_i$ from $X$ is $\Delta_i \coloneqq \min _{j, j \neq i}\left(x_i^\top x_i- x_i^\top {x}_j\right)= x_i^\top x_i-\max_{j} x_i^\top {x}_j$. 
We say $X$ is $\Delta$-separated if $\Delta_i \geq \Delta$ for all $i \in \{1,\cdots,l_X\}$. A data $D$ is $\Delta$-separated if all $X$ in $D$ are $\Delta$-separated. 
 \end{definition}

 Intuitively, $\Delta$-separated captures the intrinsic difficulty of adversarial inference on the data $D$: the less separating the data, i.e., a smaller $\Delta$, the harder to distinguish its tokens. In the context of LLMs, $\Delta$ is determined by the choice of the embedding modules. We are now ready for Lemma~\ref{lemma:ami_attn} about the advantage of $\mathcal{A}_{\mathsf{Attn}}$:
 
     \begin{lemma} \label{lemma:ami_attn}
    For a $\Delta$-separated data $D$ with i.i.d tokens of  $\mathsf{Exp}^{\textup{AMI}}$, and for any $\beta > 0$ large enough such that:
             \begin{align}
                \Delta \geq {2}/{(\beta l_X)} +  \log (2 (l_X -1) l_X \beta M^2) /\beta\label{eq:delta_cond_first}
            \end{align}
            the advantage of $\mathcal{A}_{\mathsf{Attn}}$ satisfies:
        \begin{align}
    \Adv^{\textup{AMI}}(\mathcal{A}_{\mathsf{Attn}})  \geq P_{\textup{proj}}^\mathcal{D}\left( \frac{1}{\beta l_X M}\right) + \nonumber\\P_{\textup{proj}}^\mathcal{D}\left( \frac{1}{\beta l_X M}\right)^{2 n l_X} - P_{\textup{box}}^\mathcal{D} (3\Bar{\Delta}) - 1 \label{eq:theorem_adv_api_first}
    \end{align}
    where $\Bar{\Delta} \coloneqq  2 M (l_X -1) \exp \left( 2/l_X-\beta  \Delta \right)$. $P_{\textup{proj}}^\mathcal{D}(\delta)$  is the probability that the projected component between two independent tokens drawn from $\mathcal{D}$ is smaller than  $\delta$ and $P_{\textup{box}}^\mathcal{D} (\delta)$ is the probability that a random token drawn from $\mathcal{D}$ is in the cube of size $2\delta$ centering at the mean of the tokens in $\mathcal{D}$. (Proof in Appx.~\ref{appx:api_attn_advantage})
    \end{lemma} 


\begin{table*}[ht]
\caption{General information of our experiments.}\label{table:exp:general}
\vspace{-2mm}
\begin{adjustbox}{width=1\textwidth}
\begin{tabular}{@{}ccccc@{}}
\toprule
\textbf{Experiment} & \textbf{No. runs}           & \textbf{Adversary} &  \textbf{Dataset} & \textbf{Language model} \\ \midrule
Fig.~\ref{fig:adv_and_lowcond}                 & 200 &     $\mathcal{A}_{\mathsf{Attn}}  $         & One-hot / Spherical / Gaussian             & No                  
\\
Fig.~\ref{fig:exp:API_no_LDP}                 & $20 \times 500$ & $\mathcal{A}_{\mathsf{Attn}}$   & IMDB             & BERT                  \\
Table~\ref{table:no_dp}                 & $4\times 3 \times 1 \times 40$ & $\mathcal{A}_{\mathsf{FC}}$ / $\mathcal{A}_{\mathsf{Attn}}$  / AMI~\citep{nguyen2023active} & IMDB / Twitter / Yelp / Finance              & BERT / DistilBERT / RoBERTa / GPT1 / GPT2                  \\
Table~\ref{table:dp}                 & $4\times 3 \times 4 \times 10$ &   $\mathcal{A}_{\mathsf{FC}}$ / $\mathcal{A}_{\mathsf{Attn}}$  / AMI~\citep{nguyen2023active}       & IMDB / Twitter / Yelp / Finance     & BERT / DistilBERT / RoBERTa / GPT1 / GPT2       \\
Table~\ref{table:layer}                &  $4 \times 1 \times 4 \times 10$ & $\mathcal{A}_{\mathsf{FC}}$ / $\mathcal{A}_{\mathsf{Attn}}$                   & IMDB / Twitter / Yelp / Finance              & BERT / DistilBERT / RoBERTa / GPT1 / GPT2
\end{tabular}
\end{adjustbox}
\end{table*}

    The key step of proving Lemma~\ref{lemma:ami_attn} is to rigorously show the configured attention layer operates as described in the attack strategy. We first bound the outputs $z_i^h$ of the layer when $x_i \neq v$ with Lemma~\ref{lemma:retrieve_bound} (Appx.~\ref{appx:api_attn_advantage}), which is a specific case of the \textit{Exponentially Small Retrieval Error} Theorem~\citep{ramsauer2021hopfield}. The bound claims $z_i^h \approx x_i$ with a probability lower-bounded by $P_{\textup{proj}}^\mathcal{D}\left( \nicefrac{1}{\beta l_X M}\right)$. This controls the false-positive error. When $x_i = v$, $z^1_i \approx \Bar{X}$ as described in the attack strategy. The false negatives happen when there exist other tokens filtered out unintentionally, i.e., they are near the center of the embedding. This probability is bounded by $P_{\textup{box}}^\mathcal{D} (3\Bar{\Delta})$.

    We now state some remarks about the advantage (\ref{eq:theorem_adv_api_first}) on different embeddings and their asymptotic behaviors.

\begin{remark}\label{remark:delta}\textbf{$\Delta$ vs. the advantage (\ref{eq:theorem_adv_api_first}).}
    A larger $\Delta$ allows a smaller $\beta$ satisfied (\ref{eq:delta_cond_first}). It makes $P_{\textup{proj}}^\mathcal{D}\left( \nicefrac{1}{\beta l_X M}\right)$, and consequently, the lower bound (\ref{eq:theorem_adv_api_first}) larger.
\end{remark}


\begin{remark}\label{remark:onehot}\textbf{Most vulnerable embedding.}
    A data resulting in a lower bound (\ref{eq:delta_cond_first}) near 1 is one-hot data. Since it has no token alignment, $\Delta $ achieves its maximum and $P_{\textup{proj}}^\mathcal{D}\left( \nicefrac{1}{\beta l_X M}\right)$ is 1. Furthermore, since there is no token at the center of one-hot, a large $\beta$ can be selected so that $P_{\textup{box}}^\mathcal{D} (3\Bar{\Delta}) = 0$ (See Fig.~\ref{fig:adv_and_lowcond} for more).
\end{remark}

\begin{figure}[!htb]
        \centering
          \includegraphics[width=.9\linewidth]{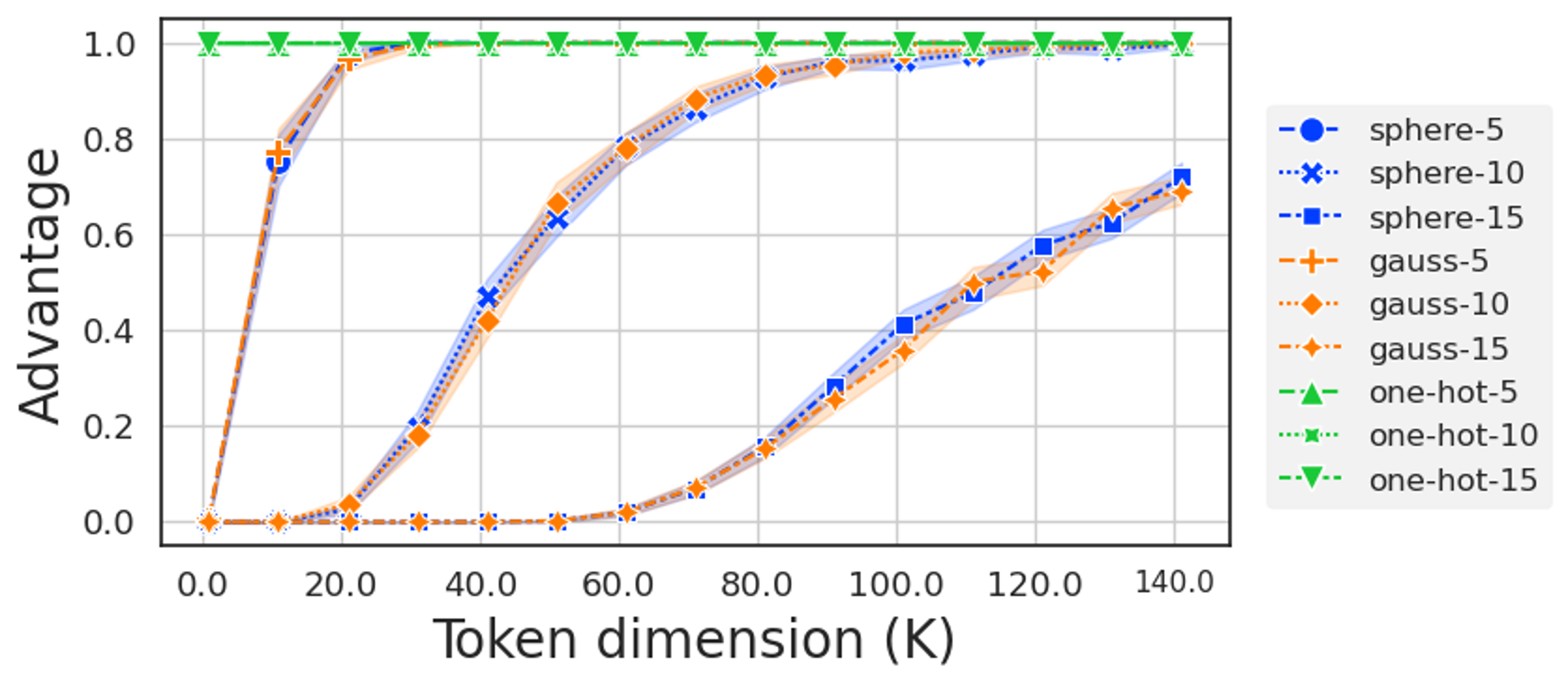}
        \caption{Simulations of the lower bound   (\ref{eq:theorem_adv_api_first}) for spherical, Gaussian and one-hot data with $l_X \in \{ 5, 10, 15  \}$ (left). $\beta$ is chosen s.t. the ratios of $\Delta $ over the RHS of (\ref{eq:delta_cond_first}) $> 1$, i.e., condition (\ref{eq:delta_cond_first}) holds (right).}
        \label{fig:adv_and_lowcond}
\end{figure}

\begin{figure}[!htb]
        \centering
          \includegraphics[width=.7\linewidth]{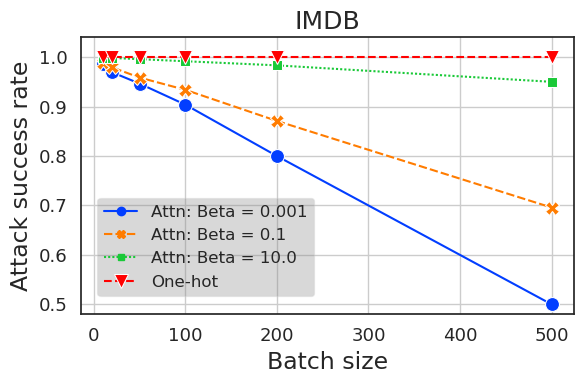}
       \caption{Success rates of $\mathcal{A}_{\mathsf{Attn}}$ on real-world IMDB dataset with different $\beta$.}
        \label{fig:exp:API_no_LDP}
\end{figure}

\begin{remark}
    \textbf{Asymptotic behavior of the advantage.}
    For high dimensional data, i.e., large $d_X$, $\Adv^{\textup{AMI}}(\mathcal{A}) \rightarrow 1$. The reason is, when $d_X \rightarrow \infty$, two random points are surely almost orthogonal ($P_{\textup{proj}}^\mathcal{D} \rightarrow 1$), and a random point is almost always at the boundary ($P_{\textup{box}}^\mathcal{D} \rightarrow 0$)~\citep{blum_hopcroft_kannan_2020}. Our Monte-Carlo simulations for the lower bound  (\ref{eq:theorem_adv_api_first}) for spherical, Gaussian, and one-hot data in Fig.~\ref{fig:adv_and_lowcond} support that claim. It is clear that one-hot data results in an advantage of 1. Appx.~\ref{appx:sphere_and_gauss} provides more explanations for other data. 
\end{remark}

\begin{table*}[ht]
\caption{Average Accuracies, F1, and AUCs of AMI attacks at different layers of LLMs on 4 real-world datasets. The methods are listed in decreasing order of trainable weights where $d_x = 768, l_x\in\{24,32\}$ and $d_h = 1000$ are the feature dimension, the number of tokens, and the method-specific parameter of the benchmark~\citep{nguyen2023active}, respectively. }\label{table:no_dp}
\centering
\resizebox{\textwidth}{!}{
\begin{tabular}{ccccccccccccccccc}
\hline
\multirow{2}{*}{\textbf{Method}} & \multirow{2}{*}{\textbf{No. params}} & \multicolumn{3}{c}{\textbf{BERT}}             & \multicolumn{3}{c}{\textbf{RoBERTa}}          & \multicolumn{3}{c}{\textbf{DistilBERT}}       & \multicolumn{3}{c}{\textbf{GPT1}}             & \multicolumn{3}{c}{\textbf{GPT2}}           \\ \cline{3-17} 
                                      &                                      & ACC           & F1            & AUC           & ACC           & F1            & AUC           & ACC           & F1            & AUC           & ACC           & F1            & AUC           & ACC           & F1            & AUC           \\ \hline
$\mathcal{A}_{\mathsf{FC-Full}}$ (Ours)                        & $2 l_x^2\times d_x^2$                & \textbf{1.00} & \textbf{1.00} & \textbf{1.00} & \textbf{1.00} & \textbf{1.00} & \textbf{1.00} & \textbf{1.00} & \textbf{1.00} & \textbf{1.00} & \textbf{1.00} & \textbf{1.00} & \textbf{1.00} & \textbf{1.00} & \textbf{1.00} & \textbf{1.00} \\
AMI FC Full                           & $l_x \times d_x \times d_{h}$   & 0.94          & 0.95          & 0.94          & 0.78& 0.74& 0.78& 0.87          & 0.88          & 0.88       & \textbf{1.00} & \textbf{1.00} & \textbf{1.00}   & 0.78          & 0.78          & 0.83\\
$\mathcal{A}_{\mathsf{Attn}}$ (Ours)                           & $20 d_x^2$                           & 0.96          & 0.95          & 0.96          & \textbf{1.00} & \textbf{1.00} & \textbf{1.00} & 0.88          & 0.88          & 0.89      & \textbf{1.00} & \textbf{1.00} & \textbf{1.00}     & 0.86          & 0.86          & 0.87\\
$\mathcal{A}_{\mathsf{FC-Token}}$ (Ours)                       & $2 d_x^2$                            & \textbf{1.00} & \textbf{1.00} & \textbf{1.00} & \textbf{1.00} & \textbf{1.00} & \textbf{1.00} & \textbf{1.00} & \textbf{1.00} & \textbf{1.00} & \textbf{1.00} & \textbf{1.00} & \textbf{1.00} & \textbf{1.00} & \textbf{1.00} & \textbf{1.00} \\
AMI FC Token                          & $d_x \times d_{h}$              & 0.94          & 0.95          & 0.94          & 0.77& 0.72& 0.77& 0.86          & 0.88          & 0.87           & \textbf{1.00} & \textbf{1.00} & \textbf{1.00}  & 0.77          & 0.77          & 0.82 \\ \hline
\end{tabular}
}
\end{table*}

\begin{table*}[ht]
\caption{Average Accuracies, F1, and AUCs of AMI attacks under DP defenses on 4 real-world datasets. The reported results are averaging among 4 DP mechanisms} \label{table:dp}
\centering
\resizebox{\textwidth}{!}{
\begin{tabular}{ccccccccccccccccc}
\hline
\multirow{2}{*}{\textbf{$\varepsilon$-DP}} & \multirow{2}{*}{\textbf{Method}} & \multicolumn{3}{c}{\textbf{BERT}} & \multicolumn{3}{c}{\textbf{RoBERTa}} & \multicolumn{3}{c}{\textbf{DistilBERT}} & \multicolumn{3}{c}{\textbf{GPT1}} & \multicolumn{3}{c}{\textbf{GPT2}} \\ \cline{3-17} 
                                        &                                       & ACC       & F1        & AUC       & ACC        & F1         & AUC        & ACC         & F1          & AUC         & ACC       & F1        & AUC       & ACC        & F1         & AUC       \\ \hline
\multirow{5}{*}{10}                     & $\mathcal{A}_{\mathsf{FC-Full}}$ (Ours)                        & 0.94      & 0.93      & 0.97      & 0.97       & 0.96       & 0.99       & 0.96        & 0.96        & 0.98    & 0.96       & 0.95       & 0.99    & 0.96      & 0.96      & 1.00            \\
                                        & AMI FC Full                           & 0.68      & 0.68      & 0.69      & 0.67& 0.66& 0.67& 0.68        & 0.69        & 0.70    & 0.62       & 0.62       & 0.65    & 0.66 & 0.65      & 0.67       \\
                                        & $\mathcal{A}_{\mathsf{Attn}}$ (Ours)                           & 0.69      & 0.65      & 0.72      & 0.76       & 0.73       & 0.81       & 0.72        & 0.68        & 0.76     & 0.81       & 0.79       & 0.86    & 0.63      & 0.59      & 0.64           \\
                                        & $\mathcal{A}_{\mathsf{FC-Token}}$ (Ours)                       & 0.87      & 0.85      & 0.92      & 0.90       & 0.89       & 0.93       & 0.90        & 0.90        & 0.96          & 0.91       & 0.89       & 0.94     & 0.91      & 0.90      & 0.94     \\
                                        & AMI FC Token                          & 0.66      & 0.67      & 0.67      & 0.65& 0.63       & 0.64       & 0.66        & 0.67        & 0.68           & 0.61       & 0.60       & 0.63    & 0.62      & 0.62      & 0.63     \\ \hline
\multirow{5}{*}{7.5}                    & $\mathcal{A}_{\mathsf{FC-Full}}$ (Ours)                        & 0.84      & 0.83      & 0.88      & 0.83       & 0.81       & 0.87       & 0.82        & 0.80        & 0.87     & 0.82       & 0.79       & 0.89    & 0.86      & 0.84      & 0.90           \\
                                        & AMI FC Full                           & 0.57      & 0.56      & 0.58      & 0.58       & 0.58       & 0.58       & 0.62        & 0.61        & 0.62     & 0.52       & 0.51       & 0.54   & 0.58& 0.59& 0.59       \\
                                        & $\mathcal{A}_{\mathsf{Attn}}$ (Ours)                           & 0.62      & 0.58      & 0.64      & 0.63       & 0.59       & 0.65       & 0.65        & 0.62        & 0.65      & 0.66       & 0.64       & 0.70    & 0.54      & 0.55      & 0.56          \\
                                        & $\mathcal{A}_{\mathsf{FC-Token}}$ (Ours)                       & 0.72      & 0.68      & 0.77      & 0.76       & 0.73       & 0.82       & 0.76        & 0.74        & 0.81       & 0.69       & 0.65       & 0.73   & 0.76      & 0.74      & 0.80          \\
                                        & AMI FC Token                          & 0.55      & 0.54      & 0.56      & 0.59& 0.58& 0.60& 0.60        & 0.59        & 0.61       & 0.50       & 0.51       & 0.52     & 0.53      & 0.54      & 0.54        \\ \hline
\multirow{5}{*}{5}                      & $\mathcal{A}_{\mathsf{FC-Full}}$ (Ours)                        & 0.67      & 0.64      & 0.69      & 0.67       & 0.64       & 0.71       & 0.63        & 0.60        & 0.65      & 0.68       & 0.65       & 0.70     & 0.66      & 0.64      & 0.71         \\
                                        & AMI FC Full                           & 0.55      & 0.54      & 0.56      & 0.53       & 0.54       & 0.53       & 0.54        & 0.52        & 0.55       & 0.53       & 0.51       & 0.53    & 0.55      & 0.56      & 0.56         \\
                                        & $\mathcal{A}_{\mathsf{Attn}}$ (Ours)                           & 0.52      & 0.50      & 0.53      & 0.54       & 0.52       & 0.54       & 0.52        & 0.50        & 0.53    & 0.56       & 0.52       & 0.56     & 0.51      & 0.51      & 0.51           \\
                                        & $\mathcal{A}_{\mathsf{FC-Token}}$ (Ours)                       & 0.57      & 0.53      & 0.60      & 0.60       & 0.58       & 0.63       & 0.59        & 0.58        & 0.61     & 0.58       & 0.53       & 0.59    & 0.57      & 0.55      & 0.59           \\
                                        & AMI FC Token                          & 0.53      & 0.52      & 0.53      & 0.51       & 0.51       & 0.50       & 0.52        & 0.50        & 0.53    & 0.50       & 0.51       & 0.51       & 0.50        & 0.50      & 0.50         \\ \hline
\end{tabular}
}
\end{table*}

\begin{remark}\label{remark:beta_mem}\textbf{Impact of $\beta$.} A larger $\beta$ in $\mathcal{A}_{\mathsf{Attn}}$ increases the memorization of the attention~\citep{ramsauer2021hopfield}. We demonstrate its impact via an experiment of BERT~\citep{BERT2019} and one-hot embeddings on the IMDB dataset. Fig.~\ref{fig:exp:API_no_LDP} shows that large $\beta$ can help  $\mathcal{A}_{\mathsf{Attn}}$ make correct inferences even when the batch size is very large (500 sentences). One-hot embedding still achieves perfect inference regardless of the batch size as stated in Remark~\ref{remark:onehot}.
\end{remark}

As  $\mathcal{A}_{\mathsf{Attn}}$ can be constructed in $\mathcal{O} (d_X^3)$, the Lemma gives us the theoretical vulnerability of unprotected data in FL with attention layer. We conclude this section with Theorem~\ref{theorem:ami_attn} about that claim.
    \begin{theorem} \label{theorem:ami_attn}
    There exists an adversary $\mathcal{A}$ exploiting a self-attention layer whose advantage satisfying (\ref{eq:theorem_adv_api_first}) with time complexity $\mathcal{O} (d_X^3)$  in the game $\mathsf{Exp}^{\textup{AMI}}$.
    \end{theorem}



\section{EXPERIMENTS}\label{sect:experiments}

\begin{table}[ht]
    \centering
    \caption{Accuracies, F1, and AUCs at different layers ($\varepsilon=10$). We highlight the entries of the layers where the metrics achieve their highest values for each attack.}
    \label{table:layer}
\centering
    \resizebox{0.99\linewidth}{!}{
    \begin{tabular}{@{}cccccccc@{}}
    \toprule
    \multirow{2}{*}{\textbf{Layer}} & \multirow{2}{*}{\textbf{Method}} & \multicolumn{3}{c}{\textbf{RoBERTa}} & \multicolumn{3}{c}{\textbf{GPT1}} \\ \cmidrule(l){3-8} 
                                    &                                   & ACC        & F1         & AUC        & ACC        & F1         & AUC       \\ \midrule
    \multirow{2}{*}{Early}          & $\mathcal{A}_{\mathsf{FC-Token}}$                                & 0.88       & 0.87       & 0.92       &\textbf{ 0.93}       & \textbf{0.93}       & \textbf{0.97 }     \\
                                    & $\mathcal{A}_{\mathsf{Attn}}$                              & 0.79       &\textbf{ 0.78}       & \textbf{0.87}       &\textbf{ 0.87}       &\textbf{ 0.86}       & \textbf{0.92}      \\ \midrule
    \multirow{2}{*}{Mid}            & $\mathcal{A}_{\mathsf{FC-Token}}$                                 & \textbf{0.92  }     & \textbf{0.91}       & \textbf{0.94 }      & 0.91       & 0.89       & 0.94      \\
                                    & $\mathcal{A}_{\mathsf{Attn}}$                              &\textbf{ 0.80}       & 0.77       & 0.86       & 0.83       & 0.80       & 0.91      \\ \midrule
    \multirow{2}{*}{Late}           & $\mathcal{A}_{\mathsf{FC-Token}}$                                 & 0.91       & 0.90       & \textbf{0.94}       & 0.88       & 0.85       & 0.90      \\
                                    & $\mathcal{A}_{\mathsf{Attn}}$                              & 0.68       & 0.65       & 0.71       & 0.73       & 0.71       & 0.75      \\ \bottomrule
    \end{tabular}
    }
\end{table}

This section provides experiments demonstrating the practical risks of leaking private data in FL with LLMs~\footnote{Our code is publicly available at \url{https://github.com/vunhatminh/FL_Attacks.git}}. In particular, we implement the FC-based $\mathcal{A}_{\mathsf{FC}}$ and self-attention-based $\mathcal{A}_{\mathsf{Attn}}$ adversaries and evaluate them in real-world setting. For the FC-based, we implement 2 versions, called $\mathcal{A}_{\mathsf{FC-Full}}$ and $\mathcal{A}_{\mathsf{FC-token}}$, which operate at a full sentence and at the token level (See Remark~\ref{remark:twoFC}). For the benchmark, we extend the previous work~\citep{nguyen2023active} into two inference attacks also at sentence and token levels, called AMI FC Full and AMI FC Token, respectively.


\textbf{Experimental Settings}. Our attacks are evaluated on 5 state-of-the-art language models: BERT~\citep{BERT2019}, RoBERTa~\citep{Liu2019RoBERTaAR}, distilBERT~\citep{sanh2019distilbert}, GPT1~\citep{gpt3}, and GPT2~\citep{gpt2}. The experiments are on 4 real-world text datasets: IMDB review~\citep{imdb_dataset}, Yelp review~\citep{zhangCharacterlevelConvolutionalNetworks2015}, Twitter-emotion~\citep{saravia-etal-2018-carer}, and Finance~\citep{Casanueva2020}. All models and datasets are from the Hugginingface database~\citep{lhoest-etal-2021-datasets}. Our experiment follows the security game in Fig.~\ref{fig:AMI}. The batch size $D$ is chosen to be 40. More information on the setting of our experiments is discussed in Appx.~\ref{appx:exp_setting}.

To realize DP mechanisms, we use Generalized Randomized Response (GRR)~\citep{dwork2006calibrating}, Google RAPPOR~\citep{erlingsson2014rappor}, Histogram encoding (HE)~\citep{THE_USENIX} and Microsoft dBitFlipPM~\citep{ding2017collecting} implemented by~\citep{Arcolezi2022}. Table~\ref{table:exp:general} provides the general information of our experiments. 
The \textit{No. runs} indicates the total number of simulated security games for each point plotted/reported in our figures/tables. For instance, the expression $4\times 3 \times 4 \times 10$ means the results are averaged over 4 datasets, at 3 layers of the model, using 4 DP mechanisms and 10 security games. {The reported privacy budget is applied once to the whole dataset and that budget is for a single communication round.}


\textbf{AMI without defense.} Table~\ref{table:no_dp} reports the average accuracies, F1 scores, and AUCs of attacks on all datasets at three locations of each model. The surfaces of attacks (Fig.~\ref{fig:petuning_2}) are the first, the middle, and the last layers (specified in more detail in Appx.~\ref{appx:exp_setting}). The results are the averages among 4 datasets. Notably, FC-based attacks consistently achieve a $100\%$ success rate, aligning with the statement in Lemma~\ref{lemma:ami_none}. Furthermore, the attention-based attack exhibits competitive performance when compared to the benchmark~\citep{nguyen2023active}. Note that, all of our methods have theoretical guarantees, and do not require neural network training. The table also includes the number of trainable weights in each attack. It serves as an indicator of the amount of fingerprints resulting from the attacks.

\textbf{AMI with DP.} We apply 4 DP mechanisms with budgets $\varepsilon\in \{5,7.5,10\}$ to evaluate the attacks in the high, medium, and low privacy regimes. Table~\ref{table:dp} reports the results averaging from 4 mechanisms, 4 datasets, and 3 locations of each model. This gives us a general idea of how attacks perform under the presence of protected noise. Results for each defense are reported separately in Appx.~\ref{appx:exp_more}.

The first observation is that, at $\varepsilon=5$, almost all attacks fail to infer the pattern, i.e. AUCs $< 0.72$; however, at mid and low privacy regimes, our attacks still achieve significant success rates. Especially, the FC-based $\mathcal{A}_{\mathsf{FC-Full}}$ and $\mathcal{A}_{\mathsf{FC-Token}}$ are always the top-2 with the highest performance. The results also show the advantages of  $\mathcal{A}_{\mathsf{FC-Full}}$ over $\mathcal{A}_{\mathsf{FC-Token}}$ since it exploits all input information. On the other hand, both neural network approaches using AMI FC degrade rapidly with the presence of noise. Since they rely on over-fitting the target pattern, it becomes more challenging to successfully train the inference networks when the input dimension becomes large in the context of LLMs.

\textbf{Inference at different layers.} It is interesting to examine the privacy leakage at different locations of LLMs. We might expect that the deeper the attacking surface, the more private the data. However, we figured out that is not necessarily the case. Table~\ref{table:layer} shows the performance of our attacks at different locations of RoBERTa and GPT1 at $\varepsilon=10$. While the claim holds for GPT1, it is not correct for RoBERTa. In fact, the layer with the highest success inference rates in RoBERTa is the mid-layer 6. Our hypothesis for the phenomenon is the tokens at that layer are more separated than those at the earlier layers. More results for other models are provided in Appx.~\ref{appx:exp_more}.

\section{CONCLUSION}\label{sect:conclude}

This work studies the formal threat models for AMI attacks with dishonest FL servers and demonstrates significant privacy threats in utilizing LLMs with FL for real-world applications. We provide evidence for the high success rates of active inference attacks, confirmed by both theoretical analysis and 
experimental evaluations. Our findings underscore the critical vulnerability of unprotected data in FL when confronted with dishonest servers. We extend our investigation to practical language models and gauge the privacy risks across different levels of DP budgets. Looking ahead, our future work will focus on identifying the prerequisites for a more secure centralized FL system and implementing these conditions effectively in practice. Furthermore, from a system perspective, we intend to explore decentralized FL protocols as a means to eliminate trust in a central server. We hope this work can serve as a stepping stone for future systematic updates and modifications of existing protocols, to make them more secure and robust for Federated LLMs.

\section*{Acknowledgments}
This material is partially supported by the National Science Foundation under the SaCT program, grant number CNS-1935923.






\bibliography{minh}
\bibliographystyle{plainnat}

\section*{Checklist}



 \begin{enumerate}

 \item For all models and algorithms presented, check if you include:
 \begin{enumerate}
   \item A clear description of the mathematical setting, assumptions, algorithm, and/or model. [Yes]
   \item An analysis of the properties and complexity (time, space, sample size) of any algorithm. [Yes]
   \item (Optional) Anonymized source code, with specification of all dependencies, including external libraries. [Yes]
 \end{enumerate}

 \item For any theoretical claim, check if you include:
 \begin{enumerate}
   \item Statements of the full set of assumptions of all theoretical results. [Yes]
   \item Complete proofs of all theoretical results. [Yes]
   \item Clear explanations of any assumptions. [Yes]     
 \end{enumerate}

 \item For all figures and tables that present empirical results, check if you include:
 \begin{enumerate}
   \item The code, data, and instructions needed to reproduce the main experimental results (either in the supplemental material or as a URL). [Yes]
   \item All the training details (e.g., data splits, hyperparameters, how they were chosen). [Yes]
    \item A clear definition of the specific measure or statistics and error bars (e.g., with respect to the random seed after running experiments multiple times). [Yes, except for error statistics in the tables. The reason is the space limit.]
    \item A description of the computing infrastructure used. (e.g., type of GPUs, internal cluster, or cloud provider). [Yes]
 \end{enumerate}

 \item If you are using existing assets (e.g., code, data, models) or curating/releasing new assets, check if you include:
 \begin{enumerate}
   \item Citations of the creator If your work uses existing assets. [Yes]
   \item The license information of the assets, if applicable. [Yes]
   \item New assets either in the supplemental material or as a URL, if applicable. [Yes]
   \item Information about consent from data providers/curators. [Yes]
   \item Discussion of sensible content if applicable, e.g., personally identifiable information or offensive content. [Yes]
 \end{enumerate}

 \item If you used crowdsourcing or conducted research with human subjects, check if you include:
 \begin{enumerate}
   \item The full text of instructions given to participants and screenshots. [Not Applicable]
   \item Descriptions of potential participant risks, with links to Institutional Review Board (IRB) approvals if applicable. [Not Applicable]
   \item The estimated hourly wage paid to participants and the total amount spent on participant compensation. [Not Applicable]
 \end{enumerate}

 \end{enumerate}

\onecolumn
\appendix
\aistatstitle{APPENDIX}

This is the appendix of our paper \textit{Analysis of Privacy Leakage in Federated Large Language Models}. Its main content and outline are as follows:
\begin{itemize}
    \item Appendix~\ref{appx:vul_ami} shows more details of our results about the FC-based adversary in FL.
    \begin{itemize}
        \item Appendix~\ref{appx:ami_adv}: the description of our FC-based adversary for AMI.
        \item Appendix~\ref{appx:ami_mlp_noldp}: the proof of Lemma~\ref{lemma:ami_none}.
    \end{itemize}
    \item Appendix~\ref{appx:vul_api} shows more details of our results about the Attention-based adversary in FL.
    \begin{itemize}
        \item Appendix~\ref{appx:api_attn_advantage}: the proof of Lemma~\ref{lemma:ami_attn}.
        \item Appendix~\ref{appx:sphere_and_gauss}: the asymptotic behavior of the advantages of our self-attention-based adversary for spherical and Gaussian data. 
    \end{itemize}
    \item Appendix~\ref{appx:exp_setting} provides the details of our experiments reported in the main manuscript.
    \begin{itemize}
        \item Appendix~\ref{appx:more_dataset}: details of the datasets.
        \item Appendix~\ref{appx:implement_adv}: details implementation of our adversaries.
    \end{itemize}
    \item Appendix~\ref{appx:exp_more} provides additional experimental results.
\end{itemize}

\newpage


\section{Vulnerability of FL to AMI Attack Exploiting FC Layers} \label{appx:vul_ami}
This appendix reports the details of our theoretical results on AMI in FL. In Appx.~\ref{appx:ami_adv}, we provide the descriptions of the proposed FC-based AMI adversary used in our analysis. Appx.~\ref{appx:ami_mlp_noldp} shows the proof of Lemma~\ref{lemma:ami_none}. 

\subsection{FC-Based Adversary for AMI in FL} \label{appx:ami_adv}

We now describe the FC-based adversary  $\mathcal{A}_{\mathsf{FC}}$ mentioned in Sect.~\ref{sect:ami_mlp}. The adversary is specified by the descriptions of its 3 components $\mathcal{A}_{\mathsf{FC-INIT}}$, $\mathcal{A}_{\mathsf{FC-ATTACK}}$ and $\mathcal{A}_{\mathsf{FC-GUESS}}$.
 
 
\textbf{AMI initialization $\mathcal{A}_{\mathsf{FC-INIT}}$:} The model specified by the adversary uses FC as its first two layers. Given an input $X \in \mathbb{R}^{d_X}$, the attacker computes $\textup{ReLU}(W_l X + b_l) = \max(0, W_l X + b_l)$ where $W_l$ is the weights and $b_l$ is the bias of layer $l$. We set the dimensions of $W_1$ and $b_1$ to $2 d_X \times d_X$ and $2 d_X$, respectively. For the second layer, the attack only considers one of its output neurons, thus, it only requires the number of columns of $W_2$ to be $2 d_X$. We use $W_2[1,:]$ and $b_2[1]$ to refer to the parameters of the row of $W_2$ and the entry of $b_2$ corresponding to that neuron. Any configurations with a higher number of parameters can adopt our proposed attack because extra parameters can simply be ignored.

\textbf{AMI attack $\mathcal{A}_{\mathsf{FC-ATTACK}}$:} The weights and biases of the first two FC layers are set as:
 \begin{align}
     W_1 \leftarrow \begin{bmatrix}
    I_{d_X} \\ 
    - I_{d_X}
    \end{bmatrix}, \quad b_1 \leftarrow \begin{bmatrix}
    - T \\ 
    T 
    \end{bmatrix}, \quad 
    W_2[1,:] \leftarrow -1_{d_X}^\top, \quad b_2[1] \leftarrow \tau
 \end{align}
 where $I_{d_X}$ is the identity matrix and $1_{d_X}$ is the one vector of size $d_X$. The hyper-parameter $\tau$ controls the total allowable distance between an input $X$ and the target $T$, which can be obtained from the distribution statistics. The pseudo-code of the attack is shown in Algo.~\ref{algo:ami_attack_fc}.

\textbf{AMI guess $\mathcal{A}_{\mathsf{FC-GUESS}}$:} In the guessing phase, the AMI server returns $1$ if the gradient of $b_2[1]$ is non-zero and returns $0$ otherwise. Algo.~\ref{algo:ami_guess_fc} shows the pseudo-code of this step. 

\SetKwInput{KwInput}{Hyper-parameters} 

\begin{algorithm}[ht]
    \SetAlgoLined
    \KwInput{$ \tau \in \mathbb{R}^+$} 
    
    \cmt{Configuring $W_1 \in \mathbb{R}^{2d_X \times d_X}$ and $b_1 \in \mathbb{R}^{2d_X }$ of the first FC}\\
    $W_1 \leftarrow \begin{bmatrix}
    I_{d_X} \\ 
    - I_{d_X}
    \end{bmatrix}, \quad b_1 \leftarrow \begin{bmatrix}
    - T \\ 
    T 
    \end{bmatrix} $

    \cmt{Configuring the first row of $W_2 \in \mathbb{R}^{d \times 2 d_X}$ and the first entry of $b_2 \in \mathbb{R}^{d }$ of the second FC}\\
        $W_2[1,:] \leftarrow -1_{2d_X}^\top, \quad b_2[1] \leftarrow \tau $
        
    \textbf{Ret} all weights and biases

\caption{$\mathcal{A}_{\mathsf{FC-ATTACK}}(T)$ exploiting fully-connected layer in AMI}
\label{algo:ami_attack_fc}
\end{algorithm}

\begin{algorithm}[ht]
    \SetAlgoLined
    
    \cmt{If the gradient of $b_2[1]$ is non-zero, returns $1$}\\
    \If{$ |\dot{\theta}(b_2[1]) | > 0$}{ \textbf{Ret} $1$}
    \textbf{Ret} $0$

\caption{$\mathcal{A}_{\mathsf{FC-GUESS}}(T,  \dot{\theta})$ exploiting fully-connected layer in AMI}
\label{algo:ami_guess_fc}
\end{algorithm}

\subsection{Proof of Lemma~\ref{lemma:ami_none} on the Advantage of the Adversary $\mathcal{A}_{\mathsf{FC}}$ in FL} \label{appx:ami_mlp_noldp}


This appendix provides the proof of Lemma~\ref{lemma:ami_none}. We restate the Lemma below:


\begin{lemma*} 
    The advantage of the adversary $\mathcal{A}_{\mathsf{FC}}$ in the security game $\mathsf{Exp}^{\textup{AMI}}$ is 1, i.e., $\Adv^{\textup{AMI}}(\mathcal{A}_{\mathsf{FC}})  = 1$.
\end{lemma*}

\begin{proof}
For the model specified by $\mathcal{A}_{\mathsf{FC}}$ as discussed in Subsection~\ref{subsect:fc_adv}, its first layer computes:
\begin{align}
    \textup{ReLU} \left( \begin{bmatrix}
    I_{d_X} \\ 
    - I_{d_X}
    \end{bmatrix} X + \begin{bmatrix}
    - T \\ 
    T 
    \end{bmatrix}\right)
    =
    \textup{ReLU} \left( \begin{bmatrix}
    X - T \\ 
    T - X 
    \end{bmatrix}\right)
\end{align}
The first row of the second layer then computes:
\begin{align}
    z_0 := & \textup{ReLU} \left( - \sum_{i=1}^{d_X}\textup{ReLU} \left( \left( x_i - t_i    \right) +  \textup{ReLU} \left( t_i - x_i   \right) \right) + \tau \right) 
    =  \max \left \{  \tau - \| X - T\|_{L_1}, 0 \right \} \label{eq:z0}
\end{align}
This implies the gradient of $b_2[1] = \tau$ is non-zero if and only if $ \tau > \| X - T\|_{L_1} $.  Thus, for a small enough $\tau$,  $T \in D$ is equivalent to a non-zero gradient. Since $\dot{\theta}(b_2[1])$ is the average of gradients of $b_2[1]$ over $D$, we have:
\begin{align}
    &\textup{If } b = 0 \Longrightarrow \  z_0 = 0, \forall X \in D \Longrightarrow \  |\dot{\theta}(b_2[1]) |= 0 \Longrightarrow \   \mathcal{A}_{\mathsf{FC-GUESS}} \textup{ returns } 0\\
    &\textup{If } b = 1 \Longrightarrow \ \left\{\begin{matrix}
 z_0 > 0 \textup{ for } X = T \in D
\\ 
z_0 = 0 \textup{ for other } X \in D 
\end{matrix}\right. \Longrightarrow \  | \dot{\theta}(b_2[1]) | > 0 \Longrightarrow \  \mathcal{A}_{\mathsf{FC-GUESS}} \textup{ returns } 1
\end{align}
Thus, the advantage of $\mathcal{A}$ is $1$. 
\end{proof}

\section{Vulnerability of FL to AMI Attack Exploiting Self-Attention Mechanism}  \label{appx:vul_api}
This appendix provides the details of our theoretical results on exploiting the self-attention mechanism in FL. We show in Appx.~\ref{appx:api_attn_advantage} the proof of Lemma~\ref{lemma:ami_attn} about the advantage of our proposed attention-based attack. We then discuss the asymptotic behavior of the advantages of $\mathcal{A}_{\mathsf{Attn}}$ for spherical and Gaussian data in Appendix~\ref{appx:sphere_and_gauss}.



\subsection{Proof of  Lemma~\ref{lemma:ami_attn} on the Advantage of the Adversary $\mathcal{A}_{\mathsf{Attn}}$} \label{appx:api_attn_advantage}

We now state a Lemma bounding the error of the self-attention layer in memorization mode. The Lemma can be considered as a specific case of Theorem 5 of \citep{ramsauer2021hopfield}. In the context of that work, they use the term \textit{for well-separated pattern} in their main manuscript to indicate the condition that the Theorem holds. In fact, the condition (\ref{eq:cond_lemma}) stated in our Lemma is a sufficient condition for that Theorem of~\citep{ramsauer2021hopfield}. For the completeness of this work, we now provide the highlight of the proof of Lemma~\ref{lemma:retrieve_bound} based on the theoretical results established in the work~\citep{ramsauer2021hopfield}.

\begin{lemma}\label{lemma:retrieve_bound}
     Given a data $X$, a constant $\alpha > 0$ large enough such that, for an $x_i \in X$:
     \begin{align}
            \Delta_i \geq \frac{2}{\alpha l_X} + \frac{1}{\alpha} \log (2 (l_X -1) l_X \alpha M^2) \label{eq:cond_lemma}
    \end{align}
    then, for any $\xi$ such that $\| \xi - x_i\| \leq  \frac{1}{\alpha l_X M}$, we have
        \begin{align*}
         \left\| x_i - X \textup{softmax} \left( \alpha X^\top \xi\right) \right\| \leq 2  M (l_X -1) \exp \left( 2/l_X-\alpha  \Delta_i \right)
     \end{align*}
\end{lemma}

\begin{proof}

    Define the sphere $S_{i} \coloneqq \{ v \textup{ such that } \| v - x_i\| \leq 1/(\alpha l_X M)\}$. We now restate and apply some results of~\citep{ramsauer2021hopfield}:

    For $\Delta_i$ satisfying (\ref{eq:cond_lemma}) and a mapping $f_\alpha$ defined as $f_\alpha(\xi) =  X \textup{softmax} \left( \alpha X^\top \xi\right) $, we have:
    \begin{itemize}
        \item The image of $S_i$ induced by $f_\alpha$ is in $S_i$, i.e., $f_\alpha$ is a mapping from $S_i$ to $S_i$ (Lemma A5~\citep{ramsauer2021hopfield}).
        \item $f_\alpha$ is a contraction mapping in $S_i$ (Lemma A6~\citep{ramsauer2021hopfield}).
        \item $f_\alpha$ has a fixed point in $S_i$ (Lemma A7~\citep{ramsauer2021hopfield}).
        \item Since $f_\alpha$ is a contraction mapping in $S_i$ and $\xi \in S_i$, we have
        \begin{align*}
        \left\| x_i - f_\alpha(\xi)\right\| \leq  2  M (l_X -1) \exp \left( -\alpha (\Delta_i - 2 \max \{ \| \xi - x_i\|, \| x_i^* - x_i\|\}\right)    
        \end{align*}
        where $x_i^*$ is a fixed point in $S_i$ (Theorem 5~\citep{ramsauer2021hopfield}).
    \end{itemize}
    Since both $x_i^*$ and $\xi$ are in $S_i$, we have $\max \{ \| \xi - x_i\|, \| x_i^* - x_i\|\} \leq   1/(\alpha l_X M)$. Therefore, we obtain
    \begin{align*}
         \left\| x_i - f_\alpha(\xi)\right\| \leq  2  M (l_X -1) \exp \left( -2/l_X - \alpha \Delta_i\right)
    \end{align*}
    Thus, we have the Lemma.
\end{proof}

Intuitively, Lemma~\ref{lemma:retrieve_bound} claims that, if we have a pattern $\xi$ near $x_i$, $X \textup{softmax} \left( \alpha X^\top \xi\right)$ is exponentially near $\xi$ as a function of $\Delta_i$. Another key remark of the Lemma is that $ \left\| x_i - X \textup{softmax} \left( \alpha X^\top \xi\right) \right\|$ exponentially approaches $0$ as the input dimension increases~\citep{ramsauer2021hopfield}. 

We now prove Lemma~\ref{lemma:ami_attn} . We restate the Lemma below.

\begin{lemma*} 
    Given a $\Delta$-separated data $D$ with i.i.d tokens of the experiment $\mathsf{Exp}^{\textup{AMI}}$, for any $\beta > 0$ large enough such that:
             \begin{align*}
                \Delta \geq {2}/{(\beta l_X)} +  \log (2 (l_X -1) l_X \beta M^2) /\beta
            \end{align*}
            the advantage of $\mathcal{A}_{\mathsf{Attn}}$ satisfies:
        \begin{align*}
    \Adv^{\textup{AMI}}(\mathcal{A}_{\mathsf{Attn}})  \geq P_{\textup{proj}}^\mathcal{D}\left( \frac{1}{\beta l_X M}\right) + P_{\textup{proj}}^\mathcal{D}\left( \frac{1}{\beta l_X M}\right)^{2 n l_X} - P_{\textup{box}}^\mathcal{D} (3\Bar{\Delta}) - 1 
    \end{align*}
    where $\Bar{\Delta} \coloneqq  2 M (l_X -1) \exp \left( 2/l_X-\beta  \Delta \right)$. $P_{\textup{proj}}^\mathcal{D}(\delta)$  is the probability that the projected component between two independent tokens drawn from $\mathcal{D}$ is smaller than  $\delta$ and $P_{\textup{box}}^\mathcal{D} (\delta)$ is the probability that a random token drawn from $\mathcal{D}$ is in the cube of size $2\delta$ centering at the arithmetic mean of the tokens in $\mathcal{D}$. 
    \end{lemma*} 
\begin{proof}
 For brevity, we first consider the following expression and related notations of the output of one attention head without the head indexing $h$:
\begin{align}
      X\textup{softmax} \left( 1/\sqrt{\datt}  X^\top {W_K}^\top W_Q X \right) \label{eq:one_head_simple}
\end{align}
Notice that we omit $W_V$ because they are all set to identity. 

We now consider the AMI adversary $\mathcal{A}_{\mathsf{Attn}}$ specified in Subsection~\ref{subsect:attn_adv}. As $W \in \mathbb{R}^{d_X \times d_X}$ (line \ref{line:init} Algo.~\ref{algo:api_attack}) is initiated randomly, it has a high probability to be non-singular even with the assignment of $v$ onto its first column (line \ref{line:setv} Algo.~\ref{algo:api_attack}). For ease of analysis, we assume $W$ has full rank. If that is not the case, we can simply re-run those 2 lines of the algorithm. For the same arguments, we also assume all $W_Q^h$ and $W_K^h$ have rank $\datt = d_X -1$.

For all heads, we have $W_K = \beta(W_Q^\top)^\dagger$ (lines \ref{line:pinv1} and \ref{line:pinv2} Algo.~\ref{algo:api_attack}). As  a consequence, $\frac{1}{\beta} W_K^\top W_Q = W_Q^\dagger W_Q$ is the projection matrix onto the column space of $W_Q^\top$. By denoting $[\xi_1,\cdots,\xi_{l_x}] = \Xi = \frac{1}{\beta} {W_K}^\top W_Q X$, we have $\xi_j$ is the projection of the token $x_j$ onto that space.

For head 1 and head 3, due to line \ref{line:setv}, we can write $W=[v,w_2,\cdots,w_{d_X}]$. From the QR factorization (line \ref{line:QR}), we have:
\begin{align*}
    &QR = [v,w_2,\cdots w_{d_X}] 
    \longrightarrow R = Q^\top [v,w_2,\cdots w_{d_X}]
\end{align*}
Since $R$ is an upper triangular matrix, $v$ is orthogonal to all rows $Q_i, i \in \{2,\cdots,d_X\}$ of $Q^\top$. Furthermore, by the assignment at line \ref{line:Q_assign}, we have the column space of $W_Q^\top$ is the linear span of $\{Q_i\}_{i=2}^{d_X}$, which are all orthogonal to $v$. Consequently, the difference between $X$ and $\Xi$ is the component of $X$ along the $v$ direction:
\begin{align}
     X - \Xi^h = [ x_1 - \xi^h_1, \cdots,x_{l_X} - \xi^h_{l_X} ] = [\bar{x}^v_1, \cdots , \bar{x}^v_{l_X}], \quad h \in \{1,3\} 
     \label{eq:proj_v}
\end{align}
where $\bar{x}^v_j$ is the component of token $x_j \in \mathbb{R}^{d_X}$ along $v$. 

For head 2 and head 4, even though we do not conduct the QR-factorization, $\frac{1}{\beta} W_K^\top W_Q $ of those heads are also project matrices, just on different column space. These spaces are also of rank $d_X-1$ and it omits one direction. By calling that direction $u$, we can write the difference between $X$ and $\Xi$ for those heads as:
\begin{align}
     X - \Xi^h = [ x_1 - \xi^h_1, \cdots,x_{l_X} - \xi^h_{l_X} ] = [\bar{x}^u_1, \cdots , \bar{x}^u_{l_X}], \quad h \in \{2,4\}
     \label{eq:proj_u}
\end{align}

We now denote $f_{\alpha}: \mathbb{R}^{d_X \times l_x} \rightarrow \mathbb{R}^{d_X \times l_x}$ as:
\begin{align*}
    \Xi' = f_{\alpha}(\Xi) = X \textup{softmax} \left(\alpha X^\top \Xi  \right)
\end{align*}
For brevity, we also abuse the notation and write $ \xi' = f_{\alpha}(\xi) = X \textup{softmax} \left(\alpha X^\top \xi  \right)$ for $\xi$ and $\xi' \in \mathbb{R}^{d_X}$.

With that, the output of the layer before ReLU can be written as:
\begin{align*}
Z =& \begin{bmatrix}
  f_{\beta}(\Xi^1) - f_{\beta}(\Xi^2) - \gamma 1^\top \\ 
  f_{\beta}(\Xi^4) - f_{\beta}(\Xi^3) - \gamma 1^\top
\end{bmatrix} =
\begin{bmatrix}
  f_{\beta}(\Xi^1) - f_{\beta}(\Xi^2) - \gamma 1^\top \\ 
  f_{\beta}(\Xi^2) - f_{\beta}(\Xi^1) - \gamma 1^\top
\end{bmatrix}
\\
=&
\begin{bmatrix}
  f_{\beta}(\xi_1^1) - f_{\beta}(\xi_1^2) - \gamma, & \cdots ,&
  f_{\beta}(\xi_{l_X}^1) - f_{\beta}(\xi_{l_X}^2) - \gamma\\ 
  f_{\beta}(\xi_1^2) - f_{\beta}(\xi_1^1) - \gamma, & \cdots ,&
  f_{\beta}(\xi_{l_X}^2) - f_{\beta}(\xi_{l_X}^1) - \gamma
\end{bmatrix}
\end{align*}
where $\gamma$ is as given in Algo.~\ref{algo:api_attack} while $\beta$ is rescaled with a factor of  $1/\sqrt{\datt}$. With the above expressions, we can see that $Z$ has non-zero entries if and only if:
\begin{align}
    \exists i \textup{ such that }  \| f_{\beta}(\xi^1_i) - f_{\beta}(\xi^2_i) \|_\infty > \gamma \label{eq:gamma}
\end{align}
Note that the usage of heads 3 and 4 is for the entries of $f_{\beta}(\xi^1_i)$ that are smaller than those in $f_{\beta}(\xi^2_i)$. The condition (\ref{eq:gamma}) can also be expressed as
\begin{align*}
     \| f_{\beta}(\Xi^1) - f_{\beta}(\Xi^2) \|_\infty > \gamma
\end{align*}


For a given token $x_i$, we now consider two cases: $x_i \neq v$ and $x_i = v$.

\textbf{Case 1.} For a token $x_i \in X$ such that $x_i \neq v$, from Lemma~\ref{lemma:retrieve_bound}, we have:
\begin{align*}
    \left\| x_i - f_{\beta}(\xi^1_i) \right \|  &\leq  2  M (l_X -1) \exp \left( 2/l_X-\beta  \Delta_i \right) \\
    \left\| x_i - f_{\beta}(\xi^2_i) \right \|  &\leq  2  M (l_X -1) \exp \left( 2/l_X-\beta  \Delta_i \right)
\end{align*}

when $\|x_i - \xi^1_i \| = \|\bar{x}^v_i\| \leq 1/(\beta l_X M)$ and $\|\bar{x}^u_i\| \leq 1/(\beta l_X M) $, respectively. Note that we need to choose a $\beta$ large enough for the Lemma to hold. We denote those events by $A_i^1$ and  $A_i^2$.
Thus, from the triangle inequality, we have:
\begin{align*}
    \| f_{\beta}(\xi^1_i) - f_{\beta}(\xi^2_i) \| 
    \leq & \left\|  f_{\beta}(\xi^1_i) - x_i \right \| + \left\| x_i - f_{\beta}(\xi^2_i) \right \| 
    \leq 4  M (l_X -1) \exp \left( 2/l_X-\beta  \Delta_i \right)  \coloneqq  2\Bar{\Delta}_i
\end{align*}
with probability of $\Prob{A_i^1 \cap A_i^2}$. Here, we denote $\Bar{\Delta}_i \coloneqq 2 M (l_X -1) \exp \left( 2/l_X-\beta  \Delta_i \right)$. We further loosen the inequality with the infinity-norm, which bounds the maximum absolute difference in the pattern's feature:
\begin{align*}
     \| f_{\beta}(\xi^1_i) - f_{\beta}(\xi^2_i) \|_\infty \leq 2 \Bar{\Delta}_i
\end{align*}
Since the data point $X$ is $\Delta$-separated, i.e., $\Delta \leq \Delta_i$, we further have:
\begin{align*}
     \| f_{\beta}(\xi^1_i) - f_{\beta}(\xi^2_i) \|_\infty \leq 2 \Bar{\Delta}
\end{align*}
where $\Bar{\Delta} \coloneqq  2 M (l_X -1) \exp \left( 2/l_X-\beta  \Delta \right)$.

We now consider the event $A^1_i$. Basically, that is the event the component of $x_i$ along $v$ is smaller than a constant determined by the data distribution $\mathcal{D}$. Furthermore, since both $v$ and $x_i$ are drawn independently from the distribution (specified in the experiment $\mathsf{Exp}^{\textup{AMI}}_{\textsc{None}}$), they can be considered as two random patterns drawn from the input distribution. Thus, the probability of $A^i_1$ is the probability that the projected component between two random tokens is smaller or equal to $\frac{1}{\beta l_X M}$. Formally, given an input distribution $\mathcal{D}$, we denote $P_{\textup{proj}}^\mathcal{D}(\delta)$ the probability that the projected component between any independent patterns drawn from $\mathcal{D}$ is less than or equal to $\delta$. We then have:
\begin{align*}
    \Prob{A_i^1 } = P_{\textup{proj}}^\mathcal{D}\left( \frac{1}{\beta l_X M}\right)
\end{align*}
by definition. Similarly, for $A_i^2$, we have:
\begin{align*}
    \Prob{A_i^2 } = \Prob{\|\bar{x}^u_i\| \leq 1/(\beta l_X M)}=P_{\textup{proj}}^\mathcal{D}\left( \frac{1}{\beta l_X M}\right)
\end{align*}
Since $v$ and $u$ are independent, we obtain:
\begin{align*}
     \Prob{A_i^1 \cap A_i^2 } \geq P_{\textup{proj}}^\mathcal{D}\left( \frac{1}{\beta l_X M}\right)^2
\end{align*}

\textbf{Case 2.} Since $v$ is orthogonal to all  $\{Q_i\}_{i=2}^{d_X}$ for the matrix $Q$ at line \ref{line:QR}, $W_Q v = 0$ and, therefore, $\xi_i = 0$. Consequently, if $x_i = v$, the output of the softmax is the one-vector scaled with a factor of $1/l_X$. We then have:
\begin{align*}
    f_{\beta}(\xi_i) = X \textup{softmax}(0) = \frac{1}{l_X} \sum_{j=1}^{l_X} x_j \coloneqq \Bar{X}
\end{align*}
Since $x_i = v$, we then obtain:
\begin{align*}
     f_{\beta}(\xi^1_i) - f_{\beta}(\xi^2_i) = 
     \Bar{X} -  f_{\beta}(\xi^2_i)  =
     \left( \Bar{X} - v \right) +
     \left( x_i - f_{\beta}(\xi^2_i) \right)
\end{align*}
Thus, from triangle inequality, we have:
\begin{align}
    &\left \| f_{\beta}(\xi^1_i) - f_{\beta}(\xi^2_i) \right \|_{\infty} \geq 
      \left \| \Bar{X} - v \right  \|_{\infty} -
     \left \| x_i - f_{\beta}(\xi^2_i)  \right \|_{\infty} \\
     \geq& \left \| \Bar{X} - v \right  \|_{\infty} - \Bar{\Delta}_i \geq  \left \| \Bar{X} - v \right  \|_{\infty} - \Bar{\Delta} \label{eq:box_bound}
\end{align}
when $A^2_i$ happens, whose probability is $P_{\textup{proj}}^\mathcal{D}\left( \frac{1}{\beta l_X M}\right)$. We now consider the probability that $\left\| f_{\beta}(\xi^1_i) - f_{ \beta}(\xi^2_i) \right\|_{\infty} > 2\Bar{\Delta}$: 
\begin{align}
    &\Prob{ \left\| f_{\beta}(\xi^1_i) - f_{\beta}(\xi^2_i) \right\|_{\infty} > 2\Bar{\Delta} } \\
    =& 
    1 - \Prob{ \left\| f_{\beta}(\xi^1_i) - f_{\beta}(\xi^2_i) \right\|_{\infty} \leq 2\Bar{\Delta} } \\
    =& 
    1 - \Prob{ \left\| f_{\beta}(\xi^1_i) - f_{\beta}(\xi^2_i) \right\|_{\infty} \leq 2\Bar{\Delta} | A_i^2} \Prob{  A_i^2 } \nonumber \\ 
    -&  \Prob{ \left\| f_{\beta}(\xi^1_i) - f_{\beta}(\xi^2_i) \right\|_{\infty} \leq 2\Bar{\Delta} | \neg A_i^2 } \Prob{ \neg A_i^2 } \\
    \geq&  1 - \Prob{ \left\| f_{\beta}(\xi^1_i) - f_{\beta}(\xi^2_i) \right\|_{\infty} \leq 2\Bar{\Delta} | A_i^2}  - \Prob{ \neg A_i^2 }
    \\
    \geq&    1 - \Prob{  \left \| \Bar{X} - v \right  \|_{\infty} - \Bar{\Delta} \leq 2\Bar{\Delta} | A_i^2 } - \Prob{ \neg A_i^2 } 
    \label{eq:use_box_bound}\\
   =&    1 -  \Prob{ \left\|  \Bar{X} - v   \right\|_{\infty} \leq 3\Bar{\Delta} } - \Prob{ \neg A_i^2 } \label{eq:drop_cond}\\
    =& P_{\textup{proj}}^\mathcal{D}\left( \frac{1}{\beta l_X M}\right) - \Prob{v \in \textup{Box}(\Bar{X}, 3\Bar{\Delta}) } \label{eq:box}
\end{align}
where $\textup{Box}(x,\delta )$ is the cube of size $2\delta$ centering at $x$. The inequality (\ref{eq:use_box_bound}) is due to (\ref{eq:box_bound}) and (\ref{eq:drop_cond}) is from the fact that $u$ is independent from $X$ and $v$. 

We now consider $\Prob{v \in \textup{Box}(\Bar{X}, 3\Bar{\Delta}) }$, which is the probability that the pattern $v$ belongs to the cube of size $ 6 \Bar{\Delta}$ around the sampled mean of the tokens in $X$. We denote $P_{\textup{box}}^\mathcal{D}(\delta)$ the probability that a random pattern drawn from $\mathcal{D}$ is in the cube of size $2\delta$ centering at the arithmetic mean of the tokens in $\mathcal{D}$. If the length $l_X$ of $X$ is large enough, we have the sampled mean $\Bar{X}$ is near the arithmetic mean of the tokens and obtain $\Prob{v \in \textup{Box}(\Bar{X}, 3\Bar{\Delta}) } \approx P_{\textup{box}}^\mathcal{D}(3\Bar{\Delta}) $.

\textbf{Back to main analysis.} From the analysis of the two cases, if $v \notin X$, we have:
\begin{align*}
\Prob{ \| f_{\beta}(\xi^1_i) - f_{\beta}(\xi^2_i) \|_\infty \leq 2 \Bar{\Delta}} \geq P_{\textup{proj}}^\mathcal{D}\left( \frac{1}{\beta l_X M}\right)^2, \quad \forall i \in \{1, \cdots, l_X \}
\end{align*}
If the tokens are independent, we then have:
\begin{align*}
    \Prob{ \| f_{\beta}(\Xi^1) - f_{\beta}(\Xi^2) \|_\infty \leq 2 \Bar{\Delta}} 
    =
    \prod_{i=1}^{l_X}\Prob{ \| f_{\beta}(\xi^1_i) - f_{\beta}(\xi^2_i) \|_\infty \leq 2 \Bar{\Delta}} \geq P_{\textup{proj}}^\mathcal{D}\left( \frac{1}{\beta l_X M}\right)^{2 l_X}
\end{align*}

Since the data points of $D$ are sampled independently, if $v$ does not appear in $D$, we then have:
\begin{align*}
     \Prob{ \| f_{\beta}(\Xi^1) - f_{\beta}(\Xi^2) \|_\infty \leq 2 \Bar{\Delta} \ \textup{ for all } X \in D} \geq  P_{\textup{proj}}^\mathcal{D}\left( \frac{1}{\beta l_X M}\right)^{2 n l_X}
\end{align*}

On the other hand, if $v \in X$, we have:
\begin{align*}
    &\exists i \in \{1, \cdots, l_X \} \textup{ such that } \Prob{ \| f_{\beta}(\xi^1_i) - f_{\beta}(\xi^2_i) \|_\infty > 2 \Bar{\Delta}} \geq   1- P_{\textup{box}}^\mathcal{D}(3\Bar{\Delta}) \\
    \Rightarrow & \Prob{ \| f_{\beta}(\Xi^1) - f_{\beta}(\Xi^2) \|_\infty > 2 \Bar{\Delta}} 
    \geq   P_{\textup{proj}}^\mathcal{D}\left( \frac{1}{\beta l_X M}\right) - P_{\textup{box}}^\mathcal{D} (3\Bar{\Delta})
\end{align*}

Thus, if pattern $v$ appears in $D$, we have:
\begin{align*}
    \Prob{ \exists X \in D \textup{ such that } \| f_{\beta}(\Xi^1) - f_{\beta}(\Xi^2) \|_\infty > 2 \Bar{\Delta}} 
    \geq   P_{\textup{proj}}^\mathcal{D}\left( \frac{1}{\beta l_X M}\right) - P_{\textup{box}}^\mathcal{D} (3\Bar{\Delta})
\end{align*}

By choosing $\gamma = 2 \Bar{\Delta}$, we have the probability that our proposed adversary wins is:
\begin{align}
    P_{W} =& \Prob{v \in D} \Prob{\|\dot{\theta}_1(W^O) \|_{\infty} > 0 | v \in  D}
    +
    \Prob{v \notin D} \Prob{\|\dot{\theta}_1(W^O) \|_{\infty} = 0 | v \notin D} \label{eq:gamma_2_delta}\\
    =& \frac{1}{2} \Prob{ \exists X \in D \textup{ such that } \| f_{\beta}(\Xi^1) - f_{\beta}(\Xi^2) \|_\infty > 2 \Bar{\Delta}| v \in  D} \nonumber \\
    + & \frac{1}{2} \Prob{ \| f_{\beta}(\Xi^1) - f_{\beta}(\Xi^2) \|_\infty \leq 2 \Bar{\Delta} \ \textup{ for all } X \in D | v \notin  D} \\
    \geq & \frac{1}{2} \left( P_{\textup{proj}}^\mathcal{D}\left( \frac{1}{\beta l_X M}\right)- P_{\textup{box}}^\mathcal{D} (3\Bar{\Delta}) \right ) +
    \frac{1}{2} P_{\textup{proj}}^\mathcal{D}\left( \frac{1}{\beta l_X M}\right)^{2 n l_X}
\end{align}

Thus, the advantage of the adversary $\mathcal{A}_\mathsf{Atnn}$ in Algo.~\ref{algo:api_attack} can be bounded by:
\begin{align}
    \Adv^{\textup{AMI}}(\mathcal{A}_\mathsf{Atnn}) = 2 P_W  - 1 \geq P_{\textup{proj}}^\mathcal{D}\left( \frac{1}{\beta l_X M}\right) + P_{\textup{proj}}^\mathcal{D}\left( \frac{1}{\beta l_X M}\right)^{2 n l_X} - P_{\textup{box}}^\mathcal{D} (3\Bar{\Delta}) - 1
    \label{eq:proj_minus_box}
\end{align}

\subsection{$\Adv^{\textup{AMI}}(\mathcal{A}_\mathsf{Atnn})$ Approaches 1 for Spherical and Gaussian Data} \label{appx:sphere_and_gauss}

This appendix provides the mathematical explanations for the asymptotic behaviors of the advantages $\Adv^{\textup{AMI}}(\mathcal{A}_\mathsf{Atnn})$ for spherical and Gaussian data (Fig.~\ref{fig:adv_and_lowcond}). 

  \textit{Spherical data.} We now consider the tokens are uniformly distributed on a unit sphere, i.e., $\|v\| = M = 1$ for all $v$. We now argue that, when $\beta$ is large enough, we have the condition (\ref{eq:delta_cond_first}) of Lemma~\ref{lemma:ami_attn} and a small $\Bar{\Delta}$ so that $P_{\textup{box}}^\mathcal{D} (3\Bar{\Delta}) = 0$. For the projecting probability $P_{\textup{proj}}$, we have the distribution of the projected component between any pair of random tokens is the distribution of any one component of a random token. The reason is the choice of the second token does not matter and we can simply select it to be the standard vector $e_i$. Therefore, the expected value of the projected component between any pair of random tokens is $1/\sqrt{d_X}$. Thus, $P_{\textup{proj}}^\mathcal{D}\left( \frac{1}{\beta l_X M}\right)$ approaches 1 as $d_X$ increases. Consequentially, $\Adv^{\textup{AMI}}(\mathcal{A}_\mathsf{Atnn}) \rightarrow 1$ as $d_X$ increases.

 \textit{Gaussian data.} When the token in $X$ have standard
normally distributed components, we have the expected value and the variance of the separation of two points are $d_X$ and $3 d_X$, respectively~\citep{ramsauer2021hopfield}. We then assume that $\Delta > C_1 d_X$ for some constant $C_1$. Since the expected norm of each token is $\sqrt{d_X}$, we further assume that $M \leq C_2 \sqrt{d_X}$ for some constant $C_2$. With that, we can select $\beta = K d_X^{-1}$ with $K$ large enough such that $\Bar{\Delta}$ is small and the condition (\ref{eq:delta_cond_first}) holds. By that choice, we have  $P_{\textup{box}}^\mathcal{D} (3\Bar{\Delta})$ is the probability that all $l_X$ normal random variables are in $[-3\Bar{\Delta}, 3\Bar{\Delta}]$, which clearly approaches $0$ as $d_X$ increases. On the other hand, from the Central Limit Theorem, we have the dot product of two random tokens converge in distribution to a one-dimensional normal random variable, i.e., $ u \cdot v / \sqrt{d_X} \rightarrow \mathcal{N}(0,1)$. Thus, as $d_X$ increases, we have $\Prob{ \frac{ | u \cdot v| }{\| v\|}  < C_3} \rightarrow 1$ for a large enough constant $C_3$. As $\frac{1}{\beta l_X M} > \frac{\sqrt{d_X}}{K l_X  C_2}$, we have  $P_{\textup{proj}}^\mathcal{D}\left( \frac{1}{\beta l_X M}\right) \rightarrow 1$.

\end{proof}

\newpage

\section{Experimental Settings} \label{appx:exp_setting}
This appendix provides the experimental details of our experiments. Our experiments are implemented using Python 3.8 and conducted on a single GPU-assisted compute node that is installed with a Linux 64-bit operating system. The allocated resources include 36 CPU cores with 2 threads per core and 60GB of RAM. The node is also equipped with 8 GPUs, with 80GB of memory per GPU.

We now provide more information on the tested dataset in Appx.~\ref{appx:more_dataset} and the details implementation of our adversaries in Appx.~\ref{appx:implement_adv}.

\subsection{Dataset and the Language Models} \label{appx:more_dataset}

In total, our experiments are conducted on 3 synthetic datasets and 4 real-world datasets. The synthetic datasets are one-hot encoded data, Spherical data (data on the boundary of a unit ball), and Gaussian data (each dimension is an independent $\mathcal{N}(0,1)$). The real-world datasets are IMDB (movie reviews)~\citep{imdb_dataset}, Yelp (general reviews)~\citep{zhangCharacterlevelConvolutionalNetworks2015}, Twitter 
 (Twitter messages with emotions)~\citep{saravia-etal-2018-carer}, and Finance (meassage with intents)~\citep{Casanueva2020}.  The real-world datasets are pre-processed with the LLMs as embedding modules to obtain the data $D$ in our threat models. For the synthetic datasets, we use a batch size of $1$ since it does not affect the asymptotic behaviors and provides better intuition for the experiment of Fig.~\ref{fig:adv_and_lowcond}. For the real-world datasets, we use a batch size of 40. 

\begin{table}[ht]
\caption{General information about the LLMs.}\label{table:exp:llm}
\centering
\begin{adjustbox}{width=0.6\textwidth}
\begin{tabular}{@{}cccccc@{}}
\toprule
                          & \textbf{BERT} & \textbf{RoBERTa} & \textbf{DistilBERT} & \textbf{GPT1} & \textbf{GPT2} \\ \midrule
\textbf{No. params}       & 110M          & 125M             & 67M                 & 120M          & 137M            \\
\textbf{No. layers}       & 12            & 12               & 6                   & 12             & 6              \\
\textbf{No. tokens} ($l_x$)      & 32            & 32               & 32                  & 32            & 32              \\
\textbf{Dimension}    ($d_x$)    & 768           & 768              & 768                 & 768           & 768             \\
\textbf{Attacking layers} & 1,6,12        & 1,6,12           & 1,3,6               & 1,6,12         & 1,3,6          \\ \bottomrule
\end{tabular}
\end{adjustbox}
\end{table}

Table~\ref{table:exp:llm} provides information about the LLMs examined in our experiments. Their implementations are provided by the Huggingface library~\citep{lhoest-etal-2021-datasets}. The rows \textit{No. tokens} and \textit{dimension} refers to the number of tokens and the embedded dimension of the tokens. The dimension is 768 at all hidden layers of all language models. The \textit{Attacking layers} are the layers that we conduct our inference attacks. They correspond to the early, middle, and late locations of attacking mentioned in Sect.~\ref{sect:experiments} of the main manuscript.

\subsection{Implementation of the Adversaries} \label{appx:implement_adv}

In our theoretical analysis of $\mathcal{A}_{\mathsf{FC}}$ and $\mathcal{A}_{\mathsf{Attn}}$, we have specified how their hyper-parameters should be chosen so that theoretical guarantees can be achieved. For convenience reference, we restate those settings here:
\begin{itemize}
    \item For Lemma~\ref{lemma:ami_none},  $\tau$ can be any small positive number such that $\tau < \| X - T\|_{L_1}$  for all $X \neq T$ in the dataset. In other words, we just need to choose a very small positive number. The argument is made at Subsect.~\ref{subsect:fc_adv}.
     \item For Lemma~\ref{lemma:ami_attn}, $\beta$ is chosen large enough such that the condition (\ref{eq:delta_cond_first}) of the Lemma holds and  $\gamma$ is set to $2 \Bar{\Delta} $. The argument is made at (\ref{eq:gamma_2_delta}) in Appx.~\ref{appx:api_attn_advantage}.
\end{itemize}

\begin{table}[ht]
\caption{Values of $\beta$ in the reported experiments. }\label{table:exp:beta}
\vspace{1mm}
\centering
\begin{adjustbox}{width=0.5\textwidth}
\begin{tabular}{@{}ccc@{}}
\toprule
\textbf{Dataset}    & \textbf{Note on $\beta$}            & $\beta$  \\ \midrule
One-hot / Spherical & $\beta$ is is set to a constant     & 10                     \\
Gaussian            & $\beta \propto  1/d_X$              & $10/d_X$                      \\
Real-world data                & The more noise, the smaller $\beta$ &  2                      \\ \bottomrule
\end{tabular}
\end{adjustbox}
\end{table}

As stated in Remark~\ref{remark:beta_mem} the parameter $\beta$ of our proposed attention-based adversary determines how much the input patterns are memorized. While increasing $\beta$ can increase the adversary's success rate of inference, it decreases the attack's performance when privacy-preserving noise is added to the data. The actual values of $\beta$ in our experiments are reported in Table~\ref{table:exp:beta}.


\begin{table*}[ht]
\caption{Average Accuracies, F1, and AUCs of AMI attacks under GRR defense~\citep{dwork2006calibrating}.} \label{table:dp_grr}
\centering
\resizebox{\textwidth}{!}{
\begin{tabular}{ccccccccccccccccc}
\hline
\multirow{2}{*}{\textbf{$\varepsilon$}} & \multirow{2}{*}{\textbf{Method}}   & \multicolumn{3}{c}{\textbf{BERT}} & \multicolumn{3}{c}{\textbf{RoBERTa}} & \multicolumn{3}{c}{\textbf{DistilBERT}} & \multicolumn{3}{c}{\textbf{GPT1}} & \multicolumn{3}{c}{\textbf{GPT2}} \\ \cline{3-17} 
                                        &                                         & ACC       & F1        & AUC       & ACC        & F1         & AUC        & ACC         & F1          & AUC         & ACC       & F1        & AUC       & ACC        & F1         & AUC       \\ \hline
\multirow{3}{*}{10}                     & $\mathcal{A}_{\mathsf{FC-Full}}$  & 1.00      & 1.00      & 1.00      & 1.00       & 1.00       & 1.00       & 1.00        & 1.00        & 1.00        & 1.00      & 1.00      & 1.00      & 1.00       & 1.00       & 1.00      \\
                                       
                                        & $\mathcal{A}_{\mathsf{Attn}}$    & 0.89      & 0.88      & 0.92      & 0.99       & 0.99       & 1.00       & 0.82        & 0.80        & 0.82      & 1.00       & 1.00       & 1.00  & 0.84      & 0.82      & 0.85            \\
                                        & $\mathcal{A}_{\mathsf{FC-Token}}$                          & 1.00      & 1.00      & 1.00      & 1.00       & 1.00       & 1.00       & 1.00        & 1.00        & 1.00        & 1.00      & 1.00      & 1.00      & 1.00       & 1.00       & 1.00          \\ \hline
\multirow{3}{*}{7.5}                    & $\mathcal{A}_{\mathsf{FC-Full}}$  & 0.99      & 0.99      & 1.00      & 1.00       & 1.00       & 1.00       & 1.00        & 1.00        & 1.00        & 0.99      & 0.99      & 1.00      & 0.99       & 0.99       & 1.00      \\
                                        & $\mathcal{A}_{\mathsf{Attn}}$    & 0.77      & 0.74      & 0.82      & 0.94       & 0.93       & 0.99       & 0.82        & 0.82        & 0.85      & 0.93       & 0.93       & 0.99    & 0.70      & 0.67      & 0.75          \\
                                        & $\mathcal{A}_{\mathsf{FC-Token}}$                         & 0.97      & 0.96      & 1.00      & 0.94       & 0.93       & 0.99       & 0.98        & 0.98        & 0.99      & 0.97       & 0.97       & 1.00  & 0.99      & 0.99      & 1.00            \\
                                       \hline
\multirow{3}{*}{5}                      & $\mathcal{A}_{\mathsf{FC-Full}}$ & 0.83      & 0.81      & 0.85      & 0.89       & 0.87       & 0.95       & 0.81        & 0.79        & 0.87         & 0.82       & 0.79       & 0.90  & 0.83      & 0.79      & 0.96          \\
                                       
                                        & $\mathcal{A}_{\mathsf{Attn}}$     & 0.56      & 0.51      & 0.60      & 0.67       & 0.62       & 0.69       & 0.61        & 0.52        & 0.64     & 0.64       & 0.60       & 0.67     & 0.50      & 0.50      & 0.50          \\
                                        & $\mathcal{A}_{\mathsf{FC-Token}}$                         & 0.71      & 0.65      & 0.79      & 0.72       & 0.70       & 0.77       & 0.74        & 0.71        & 0.76    & 0.71       & 0.65       & 0.73      & 0.72      & 0.71      & 0.78          \\
                                        \hline
\end{tabular}
}
\end{table*}

\begin{table*}[ht]
\caption{Average Accuracies, F1, and AUCs of AMI attacks under RAPPOR defense~\citep{erlingsson2014rappor}.} \label{table:dp_rappor}
\centering
\resizebox{\textwidth}{!}{
\begin{tabular}{ccccccccccccccccc}
\hline
\multirow{2}{*}{\textbf{$\varepsilon$}} & \multirow{2}{*}{\textbf{Method}} & \multicolumn{3}{c}{\textbf{BERT}} & \multicolumn{3}{c}{\textbf{RoBERTa}} & \multicolumn{3}{c}{\textbf{DistilBERT}} & \multicolumn{3}{c}{\textbf{GPT1}} & \multicolumn{3}{c}{\textbf{GPT2}} \\ \cline{3-17} 
                                        &                                       & ACC       & F1        & AUC       & ACC        & F1         & AUC        & ACC         & F1          & AUC         & ACC       & F1        & AUC       & ACC        & F1         & AUC       \\ \hline
\multirow{3}{*}{10}                     & $\mathcal{A}_{\mathsf{FC-Full}}$      & 0.97      & 0.97      & 1.00      & 0.98       & 0.98       & 1.00       & 0.98        & 0.98        & 1.00         & 1.00       & 1.00       & 1.00  & 0.99      & 0.99      & 1.00          \\
                                        & $\mathcal{A}_{\mathsf{Attn}}$         & 0.63      & 0.59      & 0.69      & 0.73       & 0.70       & 0.81       & 0.71        & 0.67        & 0.74        & 0.80       & 0.77       & 0.89   & 0.55      & 0.50      & 0.54         \\
                                        & $\mathcal{A}_{\mathsf{FC-Token}}$     & 0.90      & 0.88      & 0.95      & 0.94       & 0.93       & 0.97       & 0.92        & 0.91        & 0.97      & 0.91       & 0.89       & 0.95    & 0.95      & 0.94      & 0.98          \\ \hline
\multirow{3}{*}{7.5}                    & $\mathcal{A}_{\mathsf{FC-Full}}$      & 0.87      & 0.85      & 0.90      & 0.84       & 0.82       & 0.87       & 0.82        & 0.79        & 0.90       & 0.83       & 0.81       & 0.92     & 0.88      & 0.85      & 0.91        \\
                                        & $\mathcal{A}_{\mathsf{Attn}}$         & 0.60      & 0.56      & 0.63      & 0.55       & 0.52       & 0.57       & 0.61        & 0.58        & 0.60     & 0.57       & 0.53       & 0.60    & 0.50      & 0.50      & 0.50           \\
                                        & $\mathcal{A}_{\mathsf{FC-Token}}$     & 0.68      & 0.63      & 0.76      & 0.76       & 0.73       & 0.82       & 0.73        & 0.70        & 0.78     & 0.59       & 0.55       & 0.64    & 0.73      & 0.71      & 0.81           \\ \hline
\multirow{3}{*}{5}                      & $\mathcal{A}_{\mathsf{FC-Full}}$      & 0.62      & 0.60      & 0.68      & 0.55       & 0.51       & 0.57       & 0.53        & 0.51        & 0.52       & 0.60       & 0.56       & 0.55   & 0.62      & 0.61      & 0.66          \\
                                        & $\mathcal{A}_{\mathsf{Attn}}$         & 0.50      & 0.50      & 0.50      & 0.50       & 0.50       & 0.50       & 0.50        & 0.50        & 0.50     & 0.50       & 0.50       & 0.50      & 0.50      & 0.50      & 0.50         \\
                                        & $\mathcal{A}_{\mathsf{FC-Token}}$     & 0.53      & 0.50      & 0.53      & 0.53       & 0.50       & 0.55       & 0.55        & 0.54        & 0.58     & 0.57       & 0.50       & 0.58      & 0.50      & 0.50      & 0.51        \\ \hline
\end{tabular}
}
\end{table*}

\begin{table*}[ht]
\caption{Average Accuracies, F1, and AUCs of AMI attacks under THE defense~\citep{THE_USENIX}.} \label{table:dp_the}
\centering
\resizebox{\textwidth}{!}{
\begin{tabular}{ccccccccccccccccc}
\hline
\multirow{2}{*}{\textbf{$\varepsilon$}} & \multirow{2}{*}{\textbf{Method}} & \multicolumn{3}{c}{\textbf{BERT}} & \multicolumn{3}{c}{\textbf{RoBERTa}} & \multicolumn{3}{c}{\textbf{DistilBERT}} & \multicolumn{3}{c}{\textbf{GPT1}} & \multicolumn{3}{c}{\textbf{GPT2}} \\ \cline{3-17} 
                                        &                                       & ACC       & F1        & AUC       & ACC        & F1         & AUC        & ACC         & F1          & AUC         & ACC       & F1        & AUC       & ACC        & F1         & AUC       \\ \hline
\multirow{3}{*}{10}                     & $\mathcal{A}_{\mathsf{FC-Full}}$      & 0.80      & 0.77      & 0.89      & 0.90       & 0.88       & 0.97       & 0.87        & 0.84        & 0.92       & 0.86       & 0.83       & 0.96 & 0.89      & 0.88      & 0.98            \\
                                        & $\mathcal{A}_{\mathsf{Attn}}$         & 0.56      & 0.51      & 0.60      & 0.58       & 0.53       & 0.63       & 0.65        & 0.59        & 0.68     & 0.65       & 0.61       & 0.68   & 0.58      & 0.53      & 0.60            \\
                                        & $\mathcal{A}_{\mathsf{FC-Token}}$     & 0.70      & 0.65      & 0.78      & 0.75       & 0.72       & 0.80       & 0.76        & 0.74        & 0.86       & 0.77       & 0.73       & 0.84   & 0.73      & 0.72      & 0.78          \\ \hline
\multirow{3}{*}{7.5}                    & $\mathcal{A}_{\mathsf{FC-Full}}$      & 0.67      & 0.63      & 0.74      & 0.61       & 0.57       & 0.64       & 0.62        & 0.58        & 0.61       & 0.62       & 0.57       & 0.73   & 0.71      & 0.68      & 0.77          \\
                                        & $\mathcal{A}_{\mathsf{Attn}}$         & 0.50      & 0.50      & 0.50      & 0.50       & 0.50       & 0.50       & 0.60        & 0.54        & 0.60      & 0.56       & 0.54       & 0.59    & 0.50      & 0.50      & 0.50          \\
                                        & $\mathcal{A}_{\mathsf{FC-Token}}$     & 0.58      & 0.53      & 0.61      & 0.60       & 0.55       & 0.65       & 0.59        & 0.57        & 0.66       & 0.59       & 0.53       & 0.63   & 0.60      & 0.59      & 0.63          \\ \hline
\multirow{3}{*}{5}                      & $\mathcal{A}_{\mathsf{FC-Full}}$      & 0.58      & 0.55      & 0.59      & 0.57       & 0.55       & 0.62       & 0.61        & 0.57        & 0.67       & 0.67       & 0.63       & 0.74   & 0.57      & 0.54      & 0.60          \\
                                        & $\mathcal{A}_{\mathsf{Attn}}$         & 0.50      & 0.50      & 0.50      & 0.51       & 0.50       & 0.50       & 0.50        & 0.50        & 0.50      & 0.56       & 0.53       & 0.56    & 0.50      & 0.50      & 0.50          \\
                                        & $\mathcal{A}_{\mathsf{FC-Token}}$     & 0.52      & 0.49      & 0.56      & 0.54       & 0.54       & 0.56       & 0.50        & 0.50        & 0.51     & 0.51       & 0.50       & 0.50    & 0.53      & 0.53      & 0.50           \\ \hline
\end{tabular}
}
\end{table*}

\begin{table*}[ht]
\caption{Average Accuracies, F1, and AUCs of AMI attacks under dBitFlipPM defense~\citep{ding2017collecting}.} \label{table:dp_dbit}
\centering
\resizebox{\textwidth}{!}{
\begin{tabular}{ccccccccccccccccc}
\hline
\multirow{2}{*}{\textbf{$\varepsilon$}} & \multirow{2}{*}{\textbf{Method}} & \multicolumn{3}{c}{\textbf{BERT}} & \multicolumn{3}{c}{\textbf{RoBERTa}} & \multicolumn{3}{c}{\textbf{DistilBERT}} & \multicolumn{3}{c}{\textbf{GPT1}} & \multicolumn{3}{c}{\textbf{GPT2}} \\ \cline{3-17} 
                                        &                                       & ACC       & F1        & AUC       & ACC        & F1         & AUC        & ACC         & F1          & AUC         & ACC       & F1        & AUC       & ACC        & F1         & AUC       \\ \hline
\multirow{3}{*}{10}                     & $\mathcal{A}_{\mathsf{FC-Full}}$      & 0.97      & 0.97      & 1.00      & 0.98       & 0.98       & 1.00       & 0.99        & 0.99        & 1.00       & 0.98       & 0.98       & 1.00  & 0.97      & 0.96      & 1.00           \\
                                        & $\mathcal{A}_{\mathsf{Attn}}$         & 0.66      & 0.64      & 0.67      & 0.73       & 0.71       & 0.81       & 0.72        & 0.68        & 0.79      & 0.80       & 0.79       & 0.86   & 0.57      & 0.54      & 0.59           \\
                                        & $\mathcal{A}_{\mathsf{FC-Token}}$     & 0.89      & 0.87      & 0.96      & 0.93       & 0.92       & 0.96       & 0.93        & 0.93        & 0.99       & 0.93       & 0.92       & 0.96   & 0.95      & 0.95      & 0.99          \\ \hline
\multirow{3}{*}{7.5}                    & $\mathcal{A}_{\mathsf{FC-Full}}$      & 0.84      & 0.83      & 0.90      & 0.88       & 0.86       & 0.95       & 0.86        & 0.85        & 0.96       & 0.84       & 0.81       & 0.90   & 0.86      & 0.83      & 0.92          \\
                                        & $\mathcal{A}_{\mathsf{Attn}}$         & 0.59      & 0.55      & 0.61      & 0.55       & 0.51       & 0.55       & 0.57        & 0.53        & 0.56        & 0.59       & 0.57       & 0.62   & 0.51      & 0.51      & 0.51         \\
                                        & $\mathcal{A}_{\mathsf{FC-Token}}$     & 0.66      & 0.59      & 0.73      & 0.74       & 0.72       & 0.80       & 0.75        & 0.72        & 0.80      & 0.61       & 0.55       & 0.65    & 0.70      & 0.67      & 0.77          \\ \hline
\multirow{3}{*}{5}                      & $\mathcal{A}_{\mathsf{FC-Full}}$      & 0.59      & 0.57      & 0.60      & 0.62       & 0.60       & 0.67       & 0.56        & 0.51        & 0.57       & 0.63       & 0.62       & 0.65   & 0.62      & 0.61      & 0.63          \\
                                        & $\mathcal{A}_{\mathsf{Attn}}$         & 0.55      & 0.50      & 0.57      & 0.51       & 0.50       & 0.52       & 0.50        & 0.50        & 0.50       & 0.55       & 0.50       & 0.54    & 0.50      & 0.50      & 0.50         \\
                                        & $\mathcal{A}_{\mathsf{FC-Token}}$     & 0.52      & 0.50      & 0.54      & 0.61       & 0.59       & 0.63       & 0.59        & 0.58        & 0.59      & 0.55       & 0.50       & 0.55    & 0.54      & 0.50      & 0.56          \\ \hline
\end{tabular}
}
\end{table*}

\section{More Experimental Results} \label{appx:exp_more}

Due to the length constraints, the main manuscript only provides the average results of different Differential Privacy (DP) mechanisms in Table~\ref{table:dp}. In the following, we present more detailed results for our proposed membership inference attacks under each DP mechanism separately.

Specifically, we display the accuracies, F1 scores, and AUCs of the attacks for four different DP mechanisms: Generalized Randomized Response (GRR)\citep{dwork2006calibrating} in Table\ref{table:dp_grr}, Google RAPPOR~\citep{erlingsson2014rappor} in Table~\ref{table:dp_rappor}, Histogram encoding (HE)\citep{THE_USENIX} in Table\ref{table:dp_the}, and Microsoft dBitFlipPM~\citep{ding2017collecting} in Table~\ref{table:dp_dbit}. The DP mechanisms are applied at the token's index level using the Multi-Freq-{LDPy} library~\citep{Arcolezi2022}. Despite all of these mechanisms providing the same theoretical DP guarantee, there are subtle differences in the performance of the defense methods.

It is unsurprising that GRR consistently yields the highest successful inference rates given its status as one of the pioneers in the field of DP. When considering the relative performance among the other mechanisms, no clear winner emerges as the outcomes vary depending on the attacking methods and the examined models.

It is also worth noting that GPT models appear to exhibit greater resilience against attention-based attacks when compared to BERT-based models. Our hypothesis is grounded in the distinctive architectural differences between the GPTs, which feature a "decoder-only" design, and the BERTs, which employ an "encoder-only" architecture.

\begin{table}[ht]
\caption{Accuracies, F1, and AUCs at different layers ($\varepsilon=10$).} \label{table:layer_full}
\centering
\resizebox{\textwidth}{!}{
\begin{tabular}{ccccccccccccccccc}
\hline
\multirow{2}{*}{Layer} & \multirow{2}{*}{\textbf{Method}} & \multicolumn{3}{c}{\textbf{BERT}} & \multicolumn{3}{c}{\textbf{RoBERTa}} & \multicolumn{3}{c}{\textbf{DistilBERT}} & \multicolumn{3}{c}{\textbf{GPT1}} & \multicolumn{3}{c}{\textbf{GPT2}} \\ \cline{3-17} 
                               &                                       & ACC       & F1        & AUC       & ACC        & F1         & AUC        & ACC         & F1          & AUC         & ACC       & F1        & AUC       & ACC        & F1         & AUC       \\ \hline
\multirow{3}{*}{Early}         & $\mathcal{A}_{\mathsf{FC-Full}}$      & 0.98      & 0.98      & 1.00      & 0.99       & 0.99       & 1.00       & 0.99        & 0.99        & 1.00      & 0.98       & 0.97       & 1.00    & 0.98      & 0.98      & 1.00         \\
                               & $\mathcal{A}_{\mathsf{Attn}}$         & 0.86      & 0.83      & 0.92      & 0.79       & 0.78       & 0.87       & 0.91        & 0.89        & 0.96        & 0.87       & 0.86       & 0.92    & 0.50       & 0.50       & 0.50       \\
                               & $\mathcal{A}_{\mathsf{FC-Token}}$     & 0.88      & 0.87      & 0.94      & 0.88       & 0.87       & 0.92       & 0.91        & 0.90        & 0.96        & 0.93       & 0.93       & 0.97& 0.89      & 0.89      & 0.92            \\ \hline
\multirow{3}{*}{Mid}           & $\mathcal{A}_{\mathsf{FC-Full}}$      & 0.97      & 0.97      & 1.00      & 0.96       & 0.96       & 1.00       & 0.96        & 0.96        & 1.00       & 0.96       & 0.96       & 1.00 & 0.96      & 0.95      & 1.00            \\
                               & $\mathcal{A}_{\mathsf{Attn}}$         & 0.66      & 0.63      & 0.68      & 0.80       & 0.77       & 0.86       & 0.80        & 0.76        & 0.86        & 0.83       & 0.80       & 0.91    & 0.63      & 0.59      & 0.65           \\
                               & $\mathcal{A}_{\mathsf{FC-Token}}$     & 0.93      & 0.92      & 0.96      & 0.92       & 0.91       & 0.94       & 0.90        & 0.89        & 0.95        & 0.91       & 0.89       & 0.94    & 0.91      & 0.89      & 0.94            \\ \hline
\multirow{3}{*}{Late}          & $\mathcal{A}_{\mathsf{FC-Full}}$      & 0.85      & 0.83      & 0.91      & 0.94       & 0.94       & 0.98       & 0.94        & 0.93        & 0.94       & 0.94       & 0.92       & 0.97   & 0.96      & 0.95      & 0.99          \\
                               & $\mathcal{A}_{\mathsf{Attn}}$         & 0.54      & 0.50      & 0.55      & 0.68       & 0.65       & 0.71       & 0.50        & 0.50        & 0.50       & 0.73       & 0.71       & 0.75  & 0.67      & 0.61      & 0.68           \\
                               & $\mathcal{A}_{\mathsf{FC-Token}}$     & 0.80      & 0.75      & 0.86      & 0.91       & 0.90       & 0.94       & 0.90        & 0.89        & 0.95        & 0.88       & 0.85       & 0.90 & 0.93      & 0.92      & 0.95            \\ \hline
\end{tabular}
}
\end{table}

Table~\ref{table:layer_full} reports the performance of our attacks at different locations of
all five language models with $\varepsilon = 10$. This table can be considered as the full version of Table~\ref{table:layer} in our main manuscript. We can see that, in general, it is more challenging to infer the target data as the attacking surfaces get deeper. However, in BERT and RoBERTa, we observe a slight increase in the successful inference rates of $\mathcal{A}_{\mathsf{FC-Token}}$ around the middle layers of the models.

\vfill

\end{document}